\documentclass[3p]{elsarticle}
\usepackage{setspace}

\usepackage{epsfig}
\usepackage{rotating}
\usepackage{amssymb}
\usepackage{xcolor}
\usepackage[T1]{fontenc}
\usepackage{lmodern}
\usepackage[utf8]{inputenc}
\newenvironment{proof}{\paragraph{Proof:}}{\hfill(Proved)}
\usepackage{booktabs}
\usepackage{multirow}
\usepackage{enumerate}
\usepackage{caption}
\usepackage{subcaption}
\usepackage{lscape,float}
\usepackage{graphicx}
\usepackage[superscript,biblabel]{cite}
\usepackage{xcolor}
\usepackage{soul}
\usepackage[toc]{appendix}

\usepackage[superscript,biblabel]{cite}
\setcitestyle{authoryear,round}
\makeatletter
\renewcommand\@biblabel[1]{#1.}
\makeatother
\usepackage{algorithm}
\usepackage[noend]{algpseudocode}

\usepackage[fleqn]{amsmath}

\usepackage{mathtools}
\newtheorem{theorem}{Theorem}
\newtheorem{result}{Result}
\newtheorem{example}{Example}

\makeatletter
\renewcommand\@biblabel[1]{#1.}
\makeatother

\newcommand{\esp}{\end{sloppypar}}
\newcommand{\be}{\begin{equation}}
\newcommand{\ee}{\end{equation}}
\newcommand{\beano}{\begin{eqnarray*}}

\newcommand{\eeano}{\end{eqnarray*}}
\newcommand{\bea}{\begin{eqnarray}}
\newcommand{\eea}{\end{eqnarray}}
\newcommand{\ba}{\begin{array}}
\newcommand{\ea}{\end{array}}

\newcommand{\bc}{\begin{center}}
\newcommand{\ec}{\end{center}}

\begin{document}
\title{Bayesian reliability acceptance sampling plan sampling plans under adaptive accelerated type-II censored competing risk data }
\author[label1]{Rathin Das\corref{cor1}}
\cortext[cor1]{Corresponding author}
\ead{rathindas65@gmail.com}
\author[label2]{Soumya Roy}
\author[label1]{Biswabrata Pradhan}
\address[label1]{SQC and OR Unit, Indian Statistical Institute, Kolkata, India}
\address[label2]{Indian Institute of Management Kozhikode, Kozhikode, Pin 673570, India}
\begin{abstract}
   In recent times, products have become increasingly complex and highly reliable, so failures typically occur after long periods of operation under normal conditions and may arise from multiple causes. This paper employs simple step-stress partial accelerated life testing (SSSPALT) within the competing risks framework to determine the Bayesian reliability acceptance sampling plan (BRASP) under type-II censoring. Elevating the stress during the life test incurs an additional cost that increases the cost of the life test. In this context, an adaptive scenario is also considered in that sampling plan. The adaptive scenario is as follows: the stress is increased after a certain time if the number of failures up to that point is less than a pre-specified number of failures. The Bayes decision function and Bayes risk are derived for the general loss function. An optimal BRASP under that adaptive SSSPALT is obtained for the quadratic loss function by minimizing Bayes risk. An algorithm is provided to determine the optimal proposed BRASP. Further, comparative studies are conducted between the proposed BRASP, the conventional non-accelerated BRASP, and the conventional accelerated BRASP under type-II censoring to evaluate the effectiveness of the proposed approach. Finally, the methodology is illustrated using real data.

\end{abstract}
\
\maketitle
\section{Introduction}

\indent
In today's world, maintaining high standards of product quality has become a strategic priority for manufacturers across sectors such as automotive, aerospace, consumer electronics, and defense systems. Failures in the field pose a risk to brand reputation and lead to excessive warranty claims, expensive recalls, and, eventually, a loss of market share. Consequently, industries rely on rigorous acceptance sampling plans (ASPs) to decide whether production lots meet predefined quality standards. Among various approaches, Bayesian ASPs (henceforth, BASPs), grounded in well-established decision-theoretic principles, have gained considerable attention in the extant literature, primarily due to their ability to incorporate available prior information and account for economic considerations, such as maximizing expected return or minimizing expected loss.\\
\indent In practice, the decision to accept or reject a lot is often based on a product's reliability or other lifetime characteristics. Sampling plans that incorporate such criteria are commonly referred to in the literature as Reliability Acceptance Sampling Plans (RASPs). However, real-world reliability testing is frequently constrained by limited resources, such as time, cost, or availability of test units. To address these limitations, censoring schemes such as type-I (fixed time) and type-II (fixed number of failures) are widely employed. A series of works—including \citet{yeh1990optimal, yeh1994bayesian, yeh1995bayesian} and \citet{chen2007bayesian, lin2008exact, lin2010corrections, prajapati2019new}—have examined Bayesian variable sampling plans under various such censoring schemes for exponential distributions. While these works utilize prior information to compute Bayes risk, they do not derive decision functions through formal posterior risk minimization. A significant advancement was introduced by \citet{lin2002bayesian} and later extended by \citet{liang2013optimal}, who presented a Bayesian decision-theoretic argument for deriving the decision function by minimizing the expected posterior loss function.

\indent Most of the afore-mentioned works, however, assume that each test unit fails due to a single cause. However, modern industrial products, characterized by ever-increasing structural complexity and operation under diverse environmental conditions, often fail due to multiple, distinct causes. For example, a vehicle component might fail due to either a surface defect or an internal flaw. In order to address the presence of multiple causes of failure, reliability engineers employ \emph{competing risk models}, which consider both the time to failure and the failure cause \citep{kalbfleisch2002statistical, lawless2011statistical}. Although classical inference under competing risks has been widely studied, applications of Bayesian Sampling Plans (BSPs) in this context remain limited. A recent contribution by \citet{prajapati2023optimal} addressed BSPs for exponential distributions under hybrid censoring with competing risks, laying the groundwork for further exploration.

\indent
In high-reliability industrial products, the mean time to failure under normal operating conditions is often extremely large. Conducting full-duration life tests in such settings is impractical due to cost and time constraints. Accelerated Life Tests (ALTs), wherein products are subjected to elevated stress levels to hasten failure, offer a viable alternative. However, ALTs depend on extrapolation models to relate accelerated and normal conditions, which may not be reliable for new products lacking sufficient historical data. To mitigate this, \emph{Partial Accelerated Life Testing (PALT)} introduces a more realistic strategy—testing some units under normal stress while others undergo higher stresses—offering a balance between realism and efficiency.

\indent
One effective PALT strategy is the \emph{Step-Stress PALT (SSPALT)}, where test units initially experience a baseline stress level \( s_0 \), which is subsequently elevated to higher levels \( s_1 < s_2 < \dots < s_k \) at pre-specified time points \( \tau_1 < \tau_2 < \dots < \tau_k \). The simplest case, called Simple Step-Stress PALT (SSSPALT), involves only one stress elevation. To connect lifetime distributions across stress levels, models such as the Cumulative Exposure Model (CEM)—originally proposed by \citet{sedyakin1966relationship} and later explored by \citet{kundu2017analysis}—are widely used. Bayesian sampling plans under SSSPALT for exponential and Weibull distributions have been examined in recent studies \citep{chen2022designing, chen2023designing, prajapati2024bayesian, chen2025designing}, though mostly under single-cause failure models.

\indent
In practice, ALTs present a trade-off: while they accelerate data collection and reduce time-related costs, increasing stress levels may introduce additional risks and costs, especially when overstressing leads to non-representative failure modes. To address this, we propose an \emph{Adaptive Bayesian Sampling Plan (BSP)} that dynamically adjusts the stress level based on real-time failure data. Specifically, after a pre-determined time \( \tau_1 \), the decision to increase stress from \( s_0 \) to \( s_1 \) depends on whether the number of failures \( d_1 \) up to \( \tau_1 \) is less than a threshold \( m \). If \( d_1 < m \), stress is elevated; otherwise, it remains unchanged. This framework generalizes both non-accelerated (when \( m = 0 \)) and fully accelerated (when \( m = n \)) plans, offering a unified structure to evaluate trade-offs and guide industrial testing strategy.

\indent
Adaptive designs are not new in reliability testing—see, for example, \citet{xiang2017designing}—but their integration into BSPs under step-stress testing with competing risks and censoring remains unexplored. From a production planning and quality assurance perspective, such adaptive strategies can lead to more economical and informative testing procedures, enabling better lot disposition decisions under real-world constraints. While classical and Bayesian inference methods under competing risks and step-stress schemes have been developed \citep{han2010inference, balakrishnan2008exact, han2014inference, samanta2019analysis}, the design of BSPs under SSALT with censoring and competing risks has received little attention. This paper aims to address this gap by developing an adaptive Bayesian Sampling Plan under SSSPALT for exponential lifetimes with type-II censoring and competing risks. We derive the Bayes decision rule under general and specific loss functions, evaluate its performance through numerical studies, and demonstrate its practical implications via a motivating example.

\indent 
The remainder of this paper is organized as follows. In Section~\ref{model}, we introduce the framework of the proposed \textit{Adaptive Accelerated Bayesian Sampling Plan (AABSP)} for the exponential distribution under a \textit{Simple Step-Stress Partial Accelerated Life Test (SSSPALT)} with competing risks and type-II censoring. We examine the monotonicity property of the Bayes decision function under a general loss function and utilize this property to derive explicit expressions for both the Bayes decision rule and the Bayes risk function in Section \ref{bayesrisk}. In Section~\ref{de}, we consider specific prior distributions—Gamma and uniform—and adopt a quadratic loss function to obtain closed-form expressions for the Bayes decision function and the corresponding Bayes risk. Section \ref{algorithm} presents an algorithm to find the optimal AABSP. In Section~\ref{num}, optimum AABSP is computed under different scenarios, and the effect of the parameters on optimum AABSP is studied. In Section~\ref{com}, real data is used to study an evaluation of the performance of the proposed AABSP when time point $\tau_1$ is fixed. In both Section \ref{num} and Section \ref{com}, a comparison study between AABSP and BRASP through a non-accelerated life test, referred to as CBSP in this paper, as well as between AABSP and CACBSP under type-II censoring for competing risk data, is given. Also, in Section \ref{com}, we demonstrate the concept and practical relevance of the proposed AABSP. Finally, Section~\ref{con} concludes the paper and discusses directions for future research, including possible extensions to other lifetime distributions, alternative loss structures, and more complex stress escalation schemes.

\section{Background Theory}\label{model}

\subsection{Competing Risks Model for SSSPALT}
\indent Consider a life test in which $n$ identical test units are randomly selected from a production lot and placed on a life test at time $t_0$ under an initial stress level $s_0$, representing normal usage conditions. At a pre-specified inspection time $\tau_1$, the stress level is increased from $s_0$ to a higher level $s_1$ to accelerate the test. This setup is commonly referred to in the literature as Simple Step-Stress Partially Accelerated Life Testing (SSSPALT). \\
\indent Suppose that each test unit may fail due to one of $J$ mutually exclusive and independent causes, often referred to in the literature as \textit{competing risks}. Let $X_{ij}$ denote the potential failure time of a unit due to the $j$-th cause under stress level $s_i$, where $j=1,\ldots,J$ and $i=0,1$. We further assume that $X_{ij}$s are independent exponentials, having the probability density function (p.d.f.) 
$$f_{ij}(x\ | \ \lambda_{ij})=\lambda_{ij} \exp(-\lambda_{ij}x),~~\text{for}~~x>0,$$
and cumulative distribution function (c.d.f.)
$$
F_{ij}(x\ | \ \lambda_{ij})=1-\exp(-\lambda_{ij} x), ~~~~x>0.
$$
\indent 
Let $T_i$ denote the actual failure time of the unit at stress level $s_i$. Since failure occurs as soon as any one of the $J$ competing risks triggers a failure, we have $T_i=\min\{X_{i1},\ldots,X_{iJ}\}$. It is easy to see that $T_i$ follows an exponential distribution, having the p.d.f.
\begin{align}
f_{T_i}(t\ | \ \boldsymbol{\lambda}_i)=\left(\sum_{j=1}^J \lambda_{ij}\right) \exp\left[-\left(\sum_{j=1}^J \lambda_{ij}\right)t\right], \label{DENS_T}
\end{align}
where $\boldsymbol{\lambda}_i=(\lambda_{i1},\ldots,\lambda_{iJ})$. Furthermore, let $C_i \in \{1,\ldots,J\}$ denote the cause of failure at stress level $s_i$. In a competing risks set up, the observed data typically consist of realizations of the random pair $(T_i, C_i)$. The joint p.d.f. of $(T_i,C_i)$ at stress level $s_i$ is given by
\begin{align}
    f_{(T_i,C_i)}(t,j\ | \boldsymbol{\lambda}_i)=\lambda_{ij}\exp\left[-\left(\sum_{j=1}^J \lambda_{ij}\right) t\right].\label{DENS_CT}
\end{align}
\subsection{Competing Risk Model under an Adaptive SSSPALT}
\indent We now consider an adaptive version of the SSSPALT scheme, where the decision to elevate stress from $s_0$ to $s_1$ is based on early failure information. Let $D_1$ be the number of failures observed at $\tau_1$ in a SSSPALT scheme.  Based on the value of $D_1$, a decision is made regarding the continuation of the test: if $D_1<m$, where $m$ is a pre-decided threshold, the stress level is increased from $s_0$ to a higher level $s_1$ to accelerate the test. Otherwise, the stress level remains unchanged at $s_0$. Regardless of the stress path followed, the life test continues until a total of $r ~(\leq n)$ failures are observed, in line with the traditional Type-II censoring scheme. Obviously, we have $m \le r \le n$. This adaptive SSSPALT generalizes the traditional step-stress model by incorporating a data-driven mechanism for stress adjustment, allowing potential savings in test resources when sufficient early failure data are available at normal usage conditions. \\
\indent The time-varying stress level $s$ experienced by a test unit at an adaptive SSSPALT is thus defined as:
\begin{align*}
    s=\begin{dcases*}
        s_0 & if $t\leq \tau_1$\\
        s_0 & if $t> \tau_1 ~~ \& ~~D_1\geq m$\\
        s_1 & if $t>\tau_1~~ \& ~~D_1<m$
    \end{dcases*}.
\end{align*}
Let $T$ and $C$ denote the actual failure time and the corresponding cause of failure at $s$. Under the Cumulative Exposure Model (see \citet{nelson2009accelerated,kundu2017analysis} for details), the c.d.f. of $T$ at the stress level $s$ is given by   \begin{align*}
    F_{T}(t\mid \boldsymbol{\theta})=
  1-\exp\left[-\left(\sum_{j=1}^J\lambda_{0j}\right) \tau_1+\left(\sum_{j=1}^J\lambda_{1j}\right) (t-\tau_1)\right],
\end{align*} 
provided $t>\tau_1$ and $D_1<m$, where $\boldsymbol{\theta}$ be the vector of the parameter. Here $\boldsymbol{\theta}=(\boldsymbol{\lambda}_0,\boldsymbol{\lambda}_1)$. Furthermore, the corresponding joint p.d.f. of $(T,C)$ is given by 
 \begin{align*}
     f_{(T,C)}(t,j\ | \ \boldsymbol{\theta})=\lambda_{1j}\exp\left[-\left(\sum_{j=1}^J\lambda_{0j}\right) \tau_1+\left(\sum_{j=1}^J\lambda_{1j}\right) (t-\tau_1)\right].
 \end{align*}
\indent It is easy to understand that the c.d.f. of $T$ and the joint p.d.f. of $(T, C)$ in the remaining cases can easily be obtained from (\ref{DENS_T}) and (\ref{DENS_CT}). For completeness, these expressions for the c.d.f. and joint p.d.f. in the general case (both stress paths) are summarized as follows. Let ${\delta}$ be the indicator function defined by
\begin{align*}
    {\delta}=\begin{cases}
        1 & D_1<m\\
        0 & D_1\geq m.
    \end{cases}
\end{align*}
Then, the c.d.f. of $T$ at the stress level at $s$ is given by
\small \begin{align*}
     F_{T}(t\mid \boldsymbol{\theta})=\begin{dcases*}
         1-\exp\left[-\left(\sum_{j=1}^J\lambda_{0j}\right)t\right]& if $t\leq \tau_1$\\
         \left[1-\exp\left\{-\left(\sum_{j=1}^J\lambda_{0j}\right)t\right\}\right]^{1-\delta}\left[ 1-\exp\left\{-\left(\sum_{j=1}^J\lambda_{0j}\right) \tau_1+\left(\sum_{j=1}^J\lambda_{1j}\right)(t-\tau_1)\right\}\right]^\delta& if $t> \tau_1$\\
     \end{dcases*}.
 \end{align*}
 Similarly, the corresponding joint p.d.f. of $(T,C)$ is given by
 \begin{align*}
     f_{(T,C)}(t,j\mid \boldsymbol{\theta})=\begin{dcases*}
         \lambda_{0j}\exp\left[-\left(\sum_{j=1}^J\lambda_{0j}\right)t\right]& if $t\leq \tau_1$\\
         \left[\lambda_{0j}\exp\left\{-\left(\sum_{j=1}^J\lambda_{0j}\right)t\right\}\right]^{1-\delta}\left[ \lambda_{1j}\exp\left\{-\left(\sum_{j=1}^J\lambda_{0j}\right) \tau_1+\left(\sum_{j=1}^J\lambda_{1j}\right) (t-\tau_1)\right\}\right]^\delta& if $t> \tau_1$,\\
     \end{dcases*}.
 \end{align*}
\normalsize\indent Let $p_{1j}$ and $p_{2j}$ denote the relative risks of failure on a test unit due to cause $j$ before and after $\tau_1$, respectively. It is easy to see that
\begin{align*}
    p_{1j}=\frac{\lambda_{1j}}{\left(\sum_{j=1}^J\lambda_{0j}\right)}
\end{align*}
and
\begin{align*}
    p_{2j}=\left[\frac{\lambda_{1j}}{\left(\sum_{j=1}^J\lambda_{0j}\right)}\right]^{1-\delta}\left[\frac{\lambda_{2j}}{\left(\sum_{j=1}^J\lambda_{1j}\right)}\right]^\delta.
\end{align*}

\subsection{Likelihood Function}
\indent Suppose $n$ identical test units are randomly selected from a production lot and placed on an adaptive SSSPALT at $t_0$ under a Type-II censoring scheme. For $j=1,\ldots,J$, let $D_{1j}$ and $D_{2j}$ be the number of failures before and after time $\tau_1$ due to risk $j$, respectively. Let the random vector $\boldsymbol{D}=(D_{11},\ldots,D_{1J},D_{21},\ldots,D_{2J})$ denote the number of failures. Furthermore, for $i=1,\ldots,r$, let $T^{i}$ and $C^i$ denote the failure time of $i$th test unit and the corresponding cause of failure. It is easy to understand that $D_{1j}=\#\{1\le i\le n: T^i<\tau_1, C^i=j\}$ and $D_{2j}=\#\{1\le i\le n: T^i>\tau_1, C^i=j\}$. Also, we have $\sum_{i}\sum_{j}D_{ij}=r$ and $D_1=\sum_{j} D_{1j}$. \\
\indent Let $T^{(i)}$  denote the $i$-th order statistic, for $i=1,\ldots,r$. Note that the lifetime of the remaining $(n-r)$ surviving units is censored at $T^{(r)}$. Then, the observed data in an adaptive SSSPALT scheme under the Type-II censoring scheme can be represented as
$$\mathcal{D}=\{T^{(1)}=t^{(1)},\ldots,T^{(r)}=t^{(r)},D_{11}=d_{11},\ldots,D_{1J}=d_{1J}, D_{21}=d_{21},\ldots,D_{2J}=d_{2J}\}.$$
\indent Toward this end, we further re-parameterize the model, for convenience. Let  $\lambda_j=\lambda_{0j}$ and $\lambda_{1j}=\lambda_{0j}\times\phi_j=\lambda_j \times\phi_j$, for $j=1,\ldots,J$. Since $\lambda_{1j}>\lambda_{0j}$, we have $\phi_j>1$ for all $j$. Thus, the  parameter vector of interest is given by $\boldsymbol{\theta}=(\lambda_1,\ldots,\lambda_J,\phi_1,\ldots,\phi_J)$. Given the observed data $\mathcal{D}$, the likelihood function can be written as 
\begin{align*}
    L( \boldsymbol{\theta}\ | \ \mathcal{D})
    \propto  &\begin{dcases}\prod_{j=1}^J\lambda_{j}^{d_{1j}}
        \exp\left[-\lambda_j\left(\sum_{i=1}^{r}t^{(i)}+(n-r)t^{(r)}\right)\right], &t^{(r)}\leq\tau_1 \\
      \prod_{j=1}^J\lambda_j^{d_{+j}} \phi_j^{\delta d_{2j}}\exp\left[-\lambda_j\left(\sum_{i=0}^{d_{1}}t^{(i)}+(n-d_{1j})\tau_1\right.\right.&\\
       +\left.\left. \phi_j^\delta\left(\sum_{i=d_{1}+1}^r\left(t^{(i)}-\tau_1\right)+(n-r)\left(t^{(r)}-\tau_1\right)\right)\right)\right], &t^{(r)}> \tau_1
    \end{dcases},
\end{align*}
where $d_{+j}=\sum_{i}d_{ij}$ and $t^{(0)}=0$.\\ 
\indent Towards this end, consider the following sufficient statistics:
$$w_1(\boldsymbol{z},\boldsymbol{d})\equiv w_1=\begin{cases}
\sum_{j=1}^{r}t^{(j)}+(n-r)t^{(r)}& \text{if } t^{(r)}\leq  \tau_1\\
 \sum_{j=0}^{d_1}t^{(j)}+(n-d_1)\tau_1  & \text{if } t^{(r)}>  \tau_1
\end{cases},$$ and $$w_2(\boldsymbol{z},\boldsymbol{d})\equiv w_2=\sum_{j=d_1+1}^r(t^{(j)}-\tau_1)+(n-r)(t^{(r)}-\tau_1).$$ Then, the likelihood reduces to
\begin{align} \label{e1} L( \boldsymbol{\theta}\ | \ \mathcal{D}')\propto
\begin{cases}
\prod_{j=1}^J\lambda_j^{d_{+j}}\phi^{\delta d_{2j}}\exp\left[-\lambda_j \left(w_1+\phi_j^\delta w_2\right)\right],
\end{cases}
\end{align}
where $\mathcal{D}'=(w_1,w_2,\boldsymbol{d})$ is the transformed observed data. In order to obtain the MLEs of $\boldsymbol{\theta}$, we can easily maximize $ L( \boldsymbol{\theta}\ | \ \mathcal{D}')$ with respect to $\boldsymbol{\theta}$.

\subsection{Bayesian Decision-Theoretic Framework for Adaptive SSSPALT}
\indent To implement an adaptive Simple Step-Stress Partially Accelerated Life Test (SSSPALT) in practice, it is essential to carefully determine the key design variables: the total sample size $n$, the number of observed failures $r$, the early decision threshold $m$, and the stress change time $\tau_1$. For convenience, we denote these design variables collectively by $\boldsymbol{q} = (n, r, m, \tau_1)$.\\
\indent Let $\mathcal{D}' = (w_1, w_2, \boldsymbol{d})$ denote the observed data under the reparameterized form described earlier. Based on $\mathcal{D}'$, the sampling plan leads to either acceptance or rejection of the production lot. We define the decision rule $a(\mathcal{D}' \mid \boldsymbol{q})$ as:
\begin{align*}
a(\mathcal{D}' \mid \boldsymbol{q}) =
\begin{cases}
1, & \text{if the lot is accepted based on the observed data } \mathcal{D}', \\
0, & \text{if the lot is rejected}.
\end{cases}
\end{align*}

The total cost associated with the life test depends on the decision outcome. We define two loss functions:
\begin{itemize}
    \item $\mathcal{L}(\mathcal{A} \mid \mathcal{D}', \boldsymbol{\theta}, \boldsymbol{q})$: the loss incurred when the lot is accepted.
    \item $\mathcal{L}(\mathcal{R} \mid \mathcal{D}', \boldsymbol{\theta}, \boldsymbol{q})$: the loss incurred when the lot is rejected.
\end{itemize}

These are given by
\begin{align*}
\mathcal{L}(\mathcal{A} \mid \mathcal{D}', \boldsymbol{\theta}, \boldsymbol{q}) &= h(\boldsymbol{\lambda}) + nC_s - (n - r)v_s + t^{(r)} C_t + \delta(n - d_1)C_a, 
\end{align*}
and
\begin{align*}
\mathcal{L}(\mathcal{R} \mid \mathcal{D}', \boldsymbol{\theta}, \boldsymbol{q}) &= C_r + nC_s - (n - r)v_s + t^{(r)} C_t + \delta(n - d_1)C_a,
\end{align*}
where $\boldsymbol{\lambda} = (\lambda_1, \ldots, \lambda_J)$ is the vector of baseline hazard rates. The function $h(\boldsymbol{\lambda})$ represents the loss associated with accepting a lot of potentially substandard quality, and has been motivated in prior works (e.g., \citealp{prajapati2023optimal}; \citealp{chen2022designing}). 

\vspace{0.5em}
Each cost component in the above expressions is motivated as follows:
\begin{itemize}
    \item \textbf{Sampling cost ($C_s$)}: Represents the unit cost of subjecting a test unit to the life test, covering technician time, test facility use, measurement equipment, and quality monitoring. The term $nC_s$ accounts for the total sampling cost.
    
    \item \textbf{Salvage value ($v_s$)}: Units that survive the test can often be reused or sold as second-grade products. Thus, the salvage value $(n - r)v_s$ is a cost offset, reducing the total loss.
    
    \item \textbf{Operating cost ($C_t$)}: The test consumes operational resources over time—such as power, personnel, and maintenance. The cost incurred is proportional to the duration of the life test, captured via $t^{(r)} C_t$, where $t^{(r)}$ is the censoring time.
    
    \item \textbf{Accelerated stress cost ($C_a$)}: Applying elevated stress (e.g., temperature, voltage) increases marginal resource usage or risk of induced damage. The term $\delta(n - d_1)C_a$ reflects this additional cost, incurred only when the stress level is elevated.
    
    
    \item \textbf{Rejection cost ($C_r$)}: Captures the fixed penalty of rejecting a lot, encompassing material disposal, rework, delivery delays, or contractual penalties.
\end{itemize}

The expected loss for a given sampling plan $\boldsymbol{q}$ and observed data $\mathcal{D}'$ is:
\begin{align*}
\psi(\boldsymbol{q} \mid \mathcal{D}', \boldsymbol{\theta}) 
&= a(\mathcal{D}' \mid \boldsymbol{q}) \mathcal{L}(\mathcal{A} \mid \mathcal{D}', \boldsymbol{\theta}, \boldsymbol{q}) 
+ (1 - a(\mathcal{D}' \mid \boldsymbol{q})) \mathcal{L}(\mathcal{R} \mid \mathcal{D}', \boldsymbol{\theta}, \boldsymbol{q}) \\
&= a(\mathcal{D}' \mid \boldsymbol{q}) h(\boldsymbol{\lambda}) + (1 - a(\mathcal{D}' \mid \boldsymbol{q})) C_r 
+ nC_s + \delta(n - d_1)C_a - (n - r)v_s + t^{(r)} C_t.
\end{align*}

Assume that the prior distribution of $\boldsymbol{\theta} = (\boldsymbol{\lambda}, \boldsymbol{\phi})$ is given by $p(\boldsymbol{\theta})$, and that $\lambda_j$ and $\phi_j$ are mutually independent for each $j = 1, \ldots, J$. Then,
\[
p(\boldsymbol{\theta}) = \prod_{j=1}^J p_{1j}(\lambda_j) p_{2j}(\phi_j).
\]

The posterior distribution of $\boldsymbol{\theta}$ given the observed data $\mathcal{D}'$ is:
\[
p(\boldsymbol{\theta} \mid \mathcal{D}') = \frac{L(\boldsymbol{\theta} \mid \mathcal{D}') p(\boldsymbol{\theta})}{p(\mathcal{D}')}, \quad \text{where } p(\mathcal{D}') = \int L(\boldsymbol{\theta} \mid \mathcal{D}') p(\boldsymbol{\theta}) \, d\boldsymbol{\theta}.
\]

The Bayes risk for the sampling plan $(\boldsymbol{q}, a)$ is defined as the expected loss over the joint distribution of the parameters and data:
\begin{align}
R_B(\boldsymbol{q}, a) &= \mathbb{E}_{\boldsymbol{\theta}} \mathbb{E}_{\mathcal{D}' \mid \boldsymbol{\theta}} \left[\psi(\boldsymbol{q} \mid \mathcal{D}', \boldsymbol{\theta}) \right] \notag \\
&= n(C_s - v_s) + rv_s + C_a \mathbb{E}_{\boldsymbol{\theta}}[\mathbb{E}[\delta(n - D_1) ]] + C_t \mathbb{E}_{\boldsymbol{\theta}}[\mathbb{E}[T^{(r)} \mid \boldsymbol{\theta}]] + R_1(\boldsymbol{q},a), \label{eq:bayesrisk}
\end{align}
where the term $R_1(\boldsymbol{q},a)$ accounts for the decision uncertainty and is given by:
\begin{align}
R_1(\boldsymbol{q}, a) 
&= \mathbb{E}_{\boldsymbol{\theta}} \mathbb{E}_{\mathcal{D}' \mid \boldsymbol{\theta}} 
\left[ a(\mathcal{D}' \mid \boldsymbol{q}) h(\boldsymbol{\lambda}) + (1 - a(\mathcal{D}' \mid \boldsymbol{q})) C_r \right] \notag \\
&= \mathbb{E}_{\boldsymbol{\lambda}}[h(\boldsymbol{\lambda})] + \mathbb{E}_{\boldsymbol{\theta}} \mathbb{E}_{\mathcal{D}' \mid \boldsymbol{\theta}} 
\left[(1 - a(\mathcal{D}' \mid \boldsymbol{q})) (C_r - h(\boldsymbol{\lambda})) \right]. \label{eq:r1}
\end{align}

The optimal Bayesian sampling plan (BSP) is then defined as the sampling plan and decision rule that minimize the Bayes risk:
\[
(\boldsymbol{q}_B, a_B) = \arg \min_{(\boldsymbol{q}, a) \in \mathcal{C}} R(\boldsymbol{q}, a),
\]
where $\mathcal{C}$ denotes the class of all admissible sampling plans under the adaptive SSSPALT framework.

\section{Derivation of Bayes Decision Function and Bayes Risk}\label{bayesrisk}
In this section, we derive the Bayes decision function that minimizes the Bayes risk $R_B(\boldsymbol{q},a)$ over all admissible decision functions. We also provide an explicit form of the Bayes risk corresponding to the sampling plan defined in \eqref{eq:bayesrisk}.
\subsection{Bayes decision function}\label{Bayes decision}
Recall that the Bayes risk $R_B(\boldsymbol{q},a)$, as defined in Equation~\eqref{eq:bayesrisk}, depends on the decision function $a(\cdot \mid \boldsymbol{q})$ only through the term $R_1(\boldsymbol{q},a)$. This term can be written as:
\begin{align*}
   R_1(\boldsymbol{q},a)  ={\mathbb{E}}_{\boldsymbol{\theta}}[h(\boldsymbol{\lambda})]+{\mathbb{E}}_{\boldsymbol{\theta}}{\mathbb{E}}_{\mathcal{D}'\ | \ \boldsymbol{\theta}} [(1-a(\mathcal{D}'\ | \ \boldsymbol{q}))(C_r-h(\boldsymbol{\lambda}))]\\
   ={\mathbb{E}}_{\boldsymbol{\theta}}[h(\boldsymbol{\lambda})]+{\mathbb{E}}_{\mathcal{D}'}{\mathbb{E}}_{\boldsymbol{\theta}\ | \ \mathcal{D}'} [(1-a(\mathcal{D}'\ | \ \boldsymbol{q}))(C_r-h(\boldsymbol{\lambda}))]
\end{align*}
where
\allowdisplaybreaks\small\begin{align*}
&{\mathbb{E}}_{\mathcal{D}'}{\mathbb{E}}_{\boldsymbol{\theta}\ | \ \mathcal{D}'} [(1-a(\mathcal{D}'\ | \ \boldsymbol{q}))(C_r-h(\boldsymbol{\lambda}))]\\
=&\mathop{\sum}_{\substack{\boldsymbol{d}\\\sum_{ij}d_{ij}=r}}~~\int_{w_1}\int_{w_2}[1-a(w_1,w_2,\boldsymbol{d}\ | \ \boldsymbol{q})]{\mathbb{E}}_{\boldsymbol{\theta}\ | \ w_1,w_2,\boldsymbol{d}}[C_r-h(\boldsymbol{\lambda})]~f_{(W_1,W_2,\boldsymbol{D})}(w_1,w_2,\boldsymbol{d})~dw_1~dw_2\\
=&\mathop{\sum}_{\substack{\boldsymbol{d}\\\sum_{ij}d_{ij}=r}}~~\int_{w_1}\int_{w_2}[1-a(w_1,w_2,\boldsymbol{d}\ | \ \boldsymbol{q})]\left\{C_r-\int_{\boldsymbol{\theta}}h(\boldsymbol{\lambda})~p( \boldsymbol{\theta}\ | \  w_1,w_2,\boldsymbol{d}) ~d\boldsymbol{\theta}\right\}~f_{(W_1,W_2,\boldsymbol{D})}(w_1,w_2,\boldsymbol{d})~dw_1~dw_2\\
=&\mathop{\sum}_{\substack{\boldsymbol{d}\\\sum_{ij}d_{ij}=r}}~~\int_{w_1}\int_{w_2}[1-a(w_1,w_2,\boldsymbol{d}\ | \ \boldsymbol{q})][C_r-\varphi(w_1,w_2,\boldsymbol{d})]~f_{(W_1,W_2,\boldsymbol{D})}(w_1,w_2,\boldsymbol{d})~dw_1~dw_2,
\end{align*}
with
\normalsize\begin{align*}
    f_{(W_1,W_2,\boldsymbol{D})}(w_1,w_2,\boldsymbol{d})=\int_{\boldsymbol{\theta}} f_{(W_1,W_2,\boldsymbol{D})\ | \ \boldsymbol{\theta}}(w_1,w_2,\boldsymbol{d})~p(\boldsymbol{\theta})~d\boldsymbol{\theta},
\end{align*} 
and 
\begin{align*}
    \varphi(w_1,w_2,\boldsymbol{d})=\int_{\boldsymbol{\theta}}h(\boldsymbol{\lambda})~p( \boldsymbol{\theta}\ | \  w_1,w_2,\boldsymbol{d}) ~d\boldsymbol{\theta}.
\end{align*}
To obtain Bayes decision, we need to minimize $R_1(\boldsymbol{q},a)$ with respect to $a(\cdot\ | \ \boldsymbol{q})$, which is equivalent to minimization of ${\mathbb{E}}_{\boldsymbol{\theta}}{\mathbb{E}}_{\mathcal{D}'\ | \ \boldsymbol{\theta}} [(1-a(\mathcal{D}'\ | \ \boldsymbol{q}))(C_r-h(\boldsymbol{\lambda}))]$ with respect to $a(\cdot\ | \ \boldsymbol{q})$.
Toward this end, we consider two cases:\\
\textbf{Case 1: } $\varphi(w_1,w_2,\boldsymbol{d})\leq C_r$:
\begin{align*}
 {\mathbb{E}}_{ \boldsymbol{\theta}\ | \ (w_1,w_2,\boldsymbol{d})} [(1-a(\mathcal{D}'\ | \ \boldsymbol{q}))(C_r-h(\boldsymbol{\lambda}))]=\begin{dcases*}0 &if $a(w_1,w_2,\boldsymbol{d}\ | \ \boldsymbol{q})=1$\\
 \geq 0   &if $a(w_1,w_2,\boldsymbol{d}\ | \ \boldsymbol{q})=0$
    \end{dcases*}.
\end{align*}
\textbf{Case 2: }$\varphi(w_1,w_2,\boldsymbol{d})> C_r$:
\begin{align*}
  {\mathbb{E}}_{ \boldsymbol{\theta}\ | \ (w_1,w_2,\boldsymbol{d})} [(1-a(\mathcal{D}'\ | \ \boldsymbol{q}))(C_r-h(\boldsymbol{\lambda}))]=\begin{dcases*}0 &if $a(w_1,w_2,\boldsymbol{d}\ | \ \boldsymbol{q})=1$\\
 < 0   &if $a(w_1,w_2,\boldsymbol{d}\ | \ \boldsymbol{q})=0$
    \end{dcases*}.
\end{align*}
Therefore, for each fixed $\boldsymbol{q}$, if we take
\begin{align*}
 a(w_1,w_2,\boldsymbol{d}\ | \ \boldsymbol{q}) =\begin{dcases*}
    1&  when  $ C_r- \varphi(w_1,w_2,\boldsymbol{d})\geq 0$\\
     0 &when $ C_r- \varphi(w_1,w_2,\boldsymbol{d})<0$,
    \end{dcases*}, 
\end{align*} 
then ${\mathbb{E}}_{\boldsymbol{\theta}}{\mathbb{E}}_{\mathcal{D}'\ | \ \boldsymbol{\theta}} [(1-a(\mathcal{D}'\ | \ \boldsymbol{q}))(C_r-h(\boldsymbol{\lambda}))]$ minimizes with respect to $a(\cdot\ | \ \boldsymbol{q})$.
Thus, for fixed $\boldsymbol{q}$, the Bayes decision function $a_B(\mathcal{D}'\ | \ \boldsymbol{q})$ is given by
\begin{align*}
    a_B(w_1,w_2,\boldsymbol{d}\ | \ \boldsymbol{q})=\begin{cases}
        1 &\text{ if } C_r-\varphi(w_1,w_2,\boldsymbol{d})\geq 0\\
        0 &\text{ otherwise}
    \end{cases}.
\end{align*}
Now, we provide an alternative form for the Bayes decision function, which is useful to calculate the Bayes risk. To develop an alternative form of the Bayes decision function, we consider the following theorem.
\begin{theorem} If function $h(\boldsymbol{\lambda})$ is increasing (decreasing) in $\lambda$, then the posterior expectation of $h(\boldsymbol{\lambda})$, given by $\phi(w_1,w_2,\boldsymbol{d})$, satisfies the following monotonicity property:
    \begin{enumerate}
\item For fixed $(w_2,\boldsymbol{d})$, $\phi(w_1,w_2,\boldsymbol{d})$ is decreasing (increasing) in $w_1$.
\item For fixed $(w_1,\boldsymbol{d})$, $\phi(w_1,w_2,\boldsymbol{d})$ is decreasing (increasing) in $w_2$.
    \end{enumerate}
\end{theorem}

\indent  The above theorem can be proved following the steps as in Theorem 1 of \citet{das2024bayesian}.\\

\indent For fixed $(w_2=0,\boldsymbol{d})$, $\phi(w_1,0,\boldsymbol{d})$ is a decreasing function in $w_1$. Therefore, there exist a unique point $c(\boldsymbol{d})$ such that
\begin{align*}
        \phi(w_1,0,\boldsymbol{d})>C_r &~~~~~ \text{for } w_1<c(\boldsymbol{d})\\
        \phi(w_1,0,\boldsymbol{d})<C_r &~~~~~ \text{for } w_1>c(\boldsymbol{d}).
\end{align*}
Since, for fixed $(w_1=c(\boldsymbol{d}),\boldsymbol{d})$, $\phi(c(\boldsymbol{d}),w_2,\boldsymbol{d})$ is decreasing function in $w_2$. Therefore
\begin{align*}
        \phi(w_1,w_2,\boldsymbol{d})<C_r &~~~~~\text{for }w_2>0 ~\text{and }  w_1>c(\boldsymbol{d}).
\end{align*}
For $0<w_1<c(\boldsymbol{d})$, for fixed $(w_1,\boldsymbol{d})$, $\phi(w_1,w_2,\boldsymbol{d})$ is decreasing function in $w_2$. Therefore, there exists a unique point $c_1(w_1,\boldsymbol{d})$ such that 
\begin{align*}
        \phi(w_1,w_2,\boldsymbol{d})<C_r &~~~~ \text{for } w_2>c_1(w_1,\boldsymbol{d})~\text{and } 0<w_1<c(\boldsymbol{d}) \\
        \phi(w_1,w_2,\boldsymbol{d})>C_r &~~~~\text{for } 0<w_2<c_1(w_1,\boldsymbol{d})~\text{and } 0<w_1<c(\boldsymbol{d}).
\end{align*} 
Since the upper limit of $w_1$ is $n\tau_1$, we define $c'(\boldsymbol{d})=\min\{c(\boldsymbol{d}),n\tau_1\}$.
Therefore, the Bayes decision function is 
\begin{align*}
    \delta_B=\begin{dcases*}
        1 & for $w_1>c'(\boldsymbol{d})$ or $0<w_1<c'(\boldsymbol{d})$ and $w_2>c_1(w_1,\boldsymbol{d})$\\
        0 & for $0<w_1<c'(\boldsymbol{d})$ and $0<w_2<c_1(w_1,\boldsymbol{d})$
    \end{dcases*}.
\end{align*}

\subsection{Bayes Risk}
We now derive the explicit Bayes risk function using the Bayes decision function developed above.  Then, (\ref{eq:r1}) can be re-written as
\allowdisplaybreaks
\small\begin{align}\label{R1d}
R_1(\boldsymbol{q},a)=&{\mathbb{E}}_{\boldsymbol{\theta}}[h(\boldsymbol{\lambda})]+\mathop{\sum}_{\substack{\boldsymbol{d}\\\sum_{ij}d_{ij}=r}}~~\int_{\boldsymbol{\theta}}[C_r-h(\boldsymbol{\lambda})]P(W_1<c'(\boldsymbol{d}), W_2<c_1(W_1,\boldsymbol{d}),\boldsymbol{D}=\boldsymbol{d}\ | \ \boldsymbol{\theta})~p(\boldsymbol{\theta})~d\boldsymbol{\theta}\nonumber\\
=&{\mathbb{E}}_{\boldsymbol{\theta}}[h(\boldsymbol{\lambda})]+\mathop{\sum}_{\substack{\boldsymbol{d}\\\sum_{ij}d_{ij}=r}}~~\int\limits_{w_2=0}^{c(\boldsymbol{d})}\int\limits_{w_1=0}^{c_1(w_2,\boldsymbol{d})}\int_{\boldsymbol{\theta}}[C_r-h(\boldsymbol{\lambda}))]f_{(W_1,W_2,\boldsymbol{d})}(w_1,w_2,\boldsymbol{d}\ | \ \boldsymbol{\theta})p(\boldsymbol{\theta})~dw_1~dw_2~d\boldsymbol{\theta}\nonumber\\  
=&{\mathbb{E}}_{\boldsymbol{\theta}}[h(\boldsymbol{\lambda})]+\mathop{\sum}_{\substack{\boldsymbol{d}\\\sum_{ij}d_{ij}=r}}~~H(\boldsymbol{d}),
\end{align}
\normalsize where 
\begin{align*}   H(\boldsymbol{d})=\int\limits_{w_2=0}^{c_1(\boldsymbol{d})}\int\limits_{w_1=0}^{c_1(w_1,\boldsymbol{d})}\int_{\boldsymbol{\theta}}[C_r-h(\boldsymbol{\lambda}))]f_{(W_1,W_2,\boldsymbol{d})}(w_1,w_2,\boldsymbol{d}\ | \ \boldsymbol{\theta})p(\boldsymbol{\theta})~dw_1~dw_2~d\boldsymbol{\theta}.\\  
\end{align*}
It is easy to understand that we further need to find he expressions of the ${\mathbb{E}}_{\boldsymbol{\theta}}[\mathbb{E}[\delta(n-D_1)\ | \ \boldsymbol{\theta}]]$, ${\mathbb{E}}_{\boldsymbol{\theta}}[\mathbb{E}[T^{(r)}\ | \ \boldsymbol{\theta}]]$ and $f_{(W_1,W_2,\boldsymbol{d})}(w_1,w_2,\boldsymbol{d}\ | \ \boldsymbol{\theta})$ in order to compute $R_B(\mathbf{q},a)$ given by (\ref{eq:bayesrisk}). These expressions are derived in \ref{appa}, \ref{appb}, and \ref{appc}, respectively.

\section{Bayes Decision Function and Bayes Risk for Quadratic Loss Function}\label{de}
\indent In this section, we consider a quadratic loss function and derive the expressions for the Bayes decision function and the corresponding Bayes risk, primarily for the purpose of illustration. As is standard in the literature, we consider Gamma priors for $\lambda_i$, with p.d.f. 
\begin{align}
p_{1i}(\lambda_i)=\frac{\beta_i^{\alpha_i}}{\Gamma(\alpha_i)}\lambda_i^{\alpha_i-1}\exp(-\beta_i\lambda_i),~~~\lambda_i>0,~\alpha_i>0,~ \beta_i>0,
\end{align}
for $i=1 \ldots J$. 
and $\phi_i$ follows uniform distribution of $(1,l_i)$  with pdf
\begin{align}
    p_{2i}(\phi_i)=\frac{1}{l_i-1}, ~~~~~1<\phi_i<l_i, ~i=1,2.
\end{align}
for $i=1 \ldots J$.
 The loss function is taken as a quadratic loss function
$h(\boldsymbol{\lambda})=a_{0}+\sum_{j=1}^J a_j\lambda_j+\substack{\sum_{i=1}^J\sum_{j=1}^J}_{i\leq j}a_{ij}\lambda_i\lambda_j$ as considered \citet{prajapati2023optimal}. For the designing parameter $\boldsymbol{q}=(n,r,m,\tau_1)$, we derive the Bayes decision function.

\begin{result}
    \begin{align*}
  &\prod_{j=1}^J\left(\int_{\lambda_j=0}^\infty\int_{\phi_j=1}^{l_j}\lambda_j^{d_{+j}+\alpha_j+p_j}\phi^{d_{2j}\delta}\exp\left[-\lambda_j \left(w_1+\phi_j^\delta w_2+\beta_j\right)\right]d\lambda_j~d\phi_j\right)\\
=&\prod_{j=1}^J\left(\int_1^{l_1}\phi_j^{\delta d_{2j}}\frac{\Gamma(\alpha_j+d_{1j}+d_{2j}+p_j)}{\left(w_1+\phi_j^\delta w_2+\beta_j\right)^{(\alpha_j+d_{1j}+d_{2j}+p_j)}}~d\phi_j\right)
    \end{align*}
    it is denoted by $H_{1}(w_1,w_2,\boldsymbol{d},\boldsymbol{p})$,
\end{result}
where $\boldsymbol{p}=(p_1,\ldots,p_J)$. If ${\delta}=0$, 
\begin{align*}
H_1(w_1,w_2,\boldsymbol{d},\boldsymbol{p})=\prod_{j=1}^J\frac{\Gamma(\alpha_j+d_{1j}+d_{2j}+p_j)(l_i-1)}{\left(w_1+w_2+\beta_j\right)^{(\alpha_j+d_{1j}+d_{2j}+p_j)}}
\end{align*}
and if ${\delta}=1$,
\begin{align*}
H_1(w_1,w_2,\boldsymbol{d},\boldsymbol{p})=&\prod_{j=1}^J\left(\int_1^{l_j}\phi_j^{d_{2j}}\frac{\Gamma(p_j+d_{1j}+d_{2j}+\alpha_j)}{(w_1+\phi_j w_2+\beta_j)^{p_j+d_{1j}+d_{2j}+\alpha_j}}~d\phi_j\right)\\
=&\prod_{j=1}^J\left(\int_1^{l_j}\phi_j^{d_{2j}}\frac{\Gamma(p_j+d_{1j}+d_{2j}+\alpha_j)}{(w_1+\beta_j)^{p_j+d_{1j}+d_{2j}+\alpha_j}(1+\phi_j w_2/(w_1+\beta_j))^{p_j+d_{1j}+d_{2j}+\alpha_j}}~d\phi_j\right)\\
=&\prod_{j=1}^J\left(\int_{w_2/(w_1+\beta_j)}^{w_2l_j/(w_1+\beta_j}\frac{\Gamma(p_j+d_{1j}+d_{2j}+\alpha_j)}{(w_1+\beta_j)^{p_j+d_{1j}+\alpha_j-1}w_2^{d_{2j}+1}}\frac{z^{d_{2j}}}{(1+z)^{p_j+d_{1j}+d_{2j}+\alpha_j}}~d\phi_j\right)\\
=&\prod_{j=1}^J\frac{\Gamma(p_j+d_{1j}+d_{2j}+\alpha_j)}{(w_1+\beta_j)^{p_j+d_{1j}+\alpha_j-1}w_2^{d_{2j}+1}}\\
&~~~~~~\times[I_{\zeta_{1j}}(d_{2j}+1,p_j+d_{1j}+\alpha_j-1)-I_{\zeta_{2j}}(d_{2j}+1,p_j+d_{1j}+\alpha_j-1)],
\end{align*}
where $\zeta_{ij}=\eta_{ij}/(1+\eta_{ij})$, for $i=1,2$ with $\eta_{1j}=w_2l_j/(w_1+\beta_j)$ and $\eta_{2j}=w_2/(w_1+\beta_j)$, $j=1,2.$ $B_\eta(m,b)$ is the incomplete beta function, $I_\eta(m,b)$ is the cdf of beta function. $B_\eta(m,b)$ and $I_\eta(m,b)$ are defined as
\begin{align*}
    B_\eta(m,b)=\int_0^\eta x^{m-1}(1-x)^{b-1} dy
\end{align*}
and 
\begin{align*}
   I_\eta(m,b)=\frac{B_\eta(m,b)}{B(m,b)} 
\end{align*}
respectively. 
Using the result, the $\varphi(\mathcal{D}')$ can be written as
\begin{align}
 \varphi(\mathcal{D}')=a_0+  \sum_{j=0}^J a_{j}\frac{H_{1}(w_1,w_2,\boldsymbol{d},\boldsymbol{p}_j)}{H_{1}(w_1,w_2,\boldsymbol{d},\boldsymbol{0})}+\underset{i\leq j}{\sum_{i=1}^J\sum_{i=1}^J}a_{ij}\frac{H_{1}(w_1,w_2,\boldsymbol{d},\boldsymbol{p}_{ij})}{H_{1}(w_1,w_2,\boldsymbol{d},\boldsymbol{0})},
\end{align}
where $\boldsymbol{p}_j$ is the $J$ component vector whose $j$-th component is 1 and other components are 0, $\boldsymbol{p}_{ij}$  is the $J$ component vector whose $i$-th and $j$-th components are 1 when $i\neq j$, 2 when $i=j$  and other components are 0. Therefore, the Bayes decision is 
\begin{align*}
    a_B(w_1,w_2,\boldsymbol{d}\ | \ \boldsymbol{q})=\begin{dcases*}
        1& if $\varphi(\mathcal{D}') \leq C_r$\\
        0 &otherwise.
    \end{dcases*}
\end{align*}

The explicit form of $R_1$ is given in the appendix \ref{appd}.
\section{Finding the optimal AABSP}\label{algorithm} We consider the following algorithm to find the optimal AABSP.
 The Bayes risk is a function of three discrete variables $n,r,m$ and one continuous variable $\tau_1$. The solution cannot be obtained analytically. We provide graph of Bayes risk with respect to $\tau_1$ for $n=4,r=3$ and $m=0,1,2,3$ in Figure \ref{fig:enter-label} for the hyperparameter $\alpha_{1}=2.21$, $\alpha_{2}=9.59$, $\beta_{1}=100$, $\beta_{2}=100$, $l_1=109.202$ and $l_2=11.716$ and the cost components  $c_t=0.3$, $c_s=0.6$, $v_s=0.3$, $c_a=0.1$, $c_r=80$.
 \begin{figure}[hbt!]
     \centering
     \includegraphics[width=0.49\linewidth]{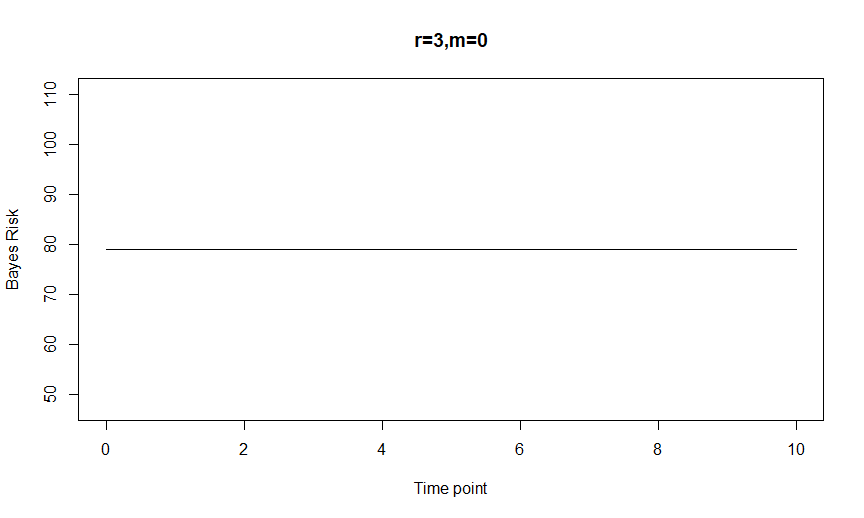} \includegraphics[width=0.49\linewidth]{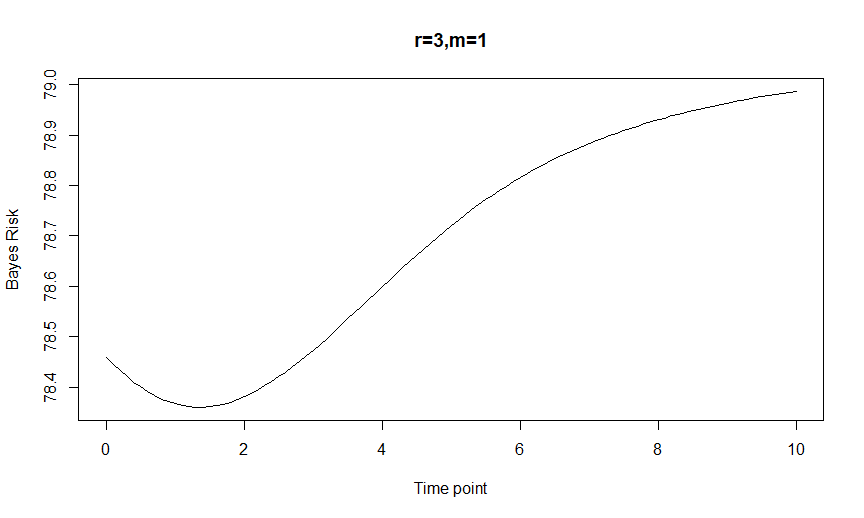} \\ \includegraphics[width=0.49\linewidth]{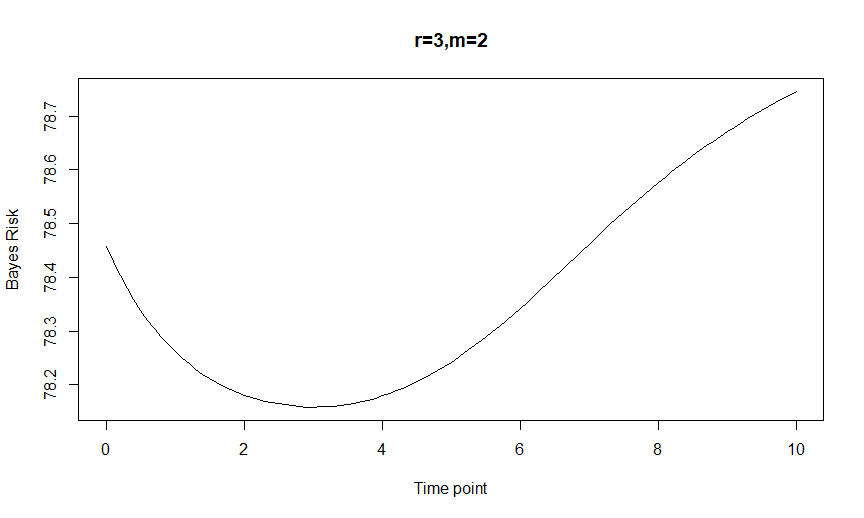} \includegraphics[width=0.49\linewidth]{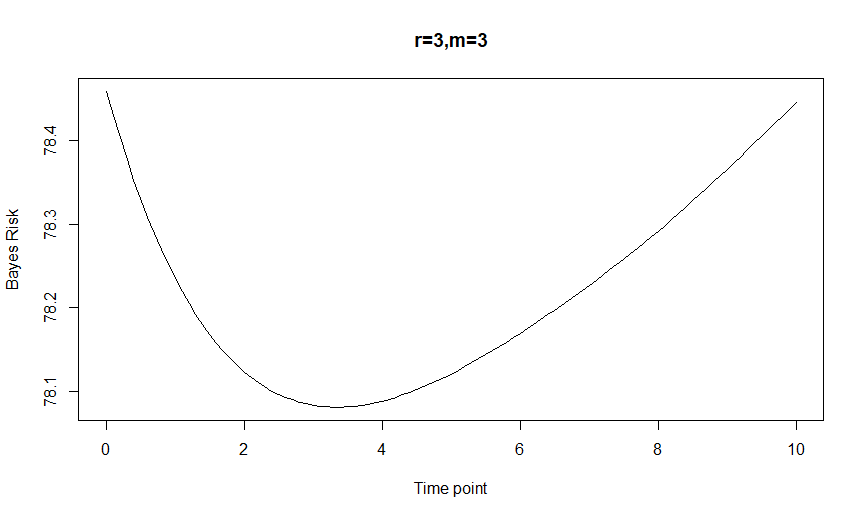}
     \caption{Graphical plot Bayes risk vs time point for different value of $m$}
     \label{fig:enter-label}
 \end{figure}
 
 In Figure \ref{fig:enter-label}, it is noted that for $m=0$, the line is parallel. This is due to the fact that for $m=0$, AABSP is converted to CBSP. In CBSP, the optimal design does not depend on the stress-changing time point $\tau_1$. So when $m=0$, we can take $\tau_1=0$. For $m>0$, in Figure \ref{fig:enter-label}, it is seen that for fixed $n,r,m$, the Bayes risk has a unique minimum. Note that the optimal value $r_B$ is less than or equal to $n_B$ and the optimal value $m_B$ is less than or equal to $r_B$. Therefore, we can say that the upper bound of $r_B$ and $m_B$ depends on the $n_B$. Next, we provide an upper bound for $n_B$.
 Let $\boldsymbol{q}_B=(n_B,t_{1B},t_{2B},m_B)$ be the optimal AABSP. Since all costs are positive and $C_s>v_s\geq 0$, all terms in $R_B(\boldsymbol{q}_B)$ are positive. So we get,
 \begin{align}\label{i1}
     R_B(\boldsymbol{q}_B,a_B)\geq n_B(C_s-v_s)
 \end{align}
 Let $(0,0,0,0)$ denote the no-sampling case, i.e., when we take the decision without life-testing. If the lot is rejected, the Bayes risk is $R(0,0,0,0)=C_r$ and if the lot is accepted, the Bayes risk is $R(0,0,0,0)={\mathbb{E}}_{\boldsymbol{\lambda}}[h(\boldsymbol{\lambda})]$ . Therefore
 $R(0,0,0,0)=\min\{{\mathbb{E}}_{\boldsymbol{\lambda}}[h(\boldsymbol{\lambda})], C_r\}$. Therefore,
 \begin{align}\label{i2}
     R_B(\boldsymbol{q}_B,a_B)\leq \min\{{\mathbb{E}}_{\boldsymbol{\lambda}}[h(\boldsymbol{\lambda})], C_r\}.
 \end{align}
 From (\ref{i1}) and (\ref{i2}), we get the upper bound of $n_B$, $$n_0=\left\lfloor\frac{\min\{{\mathbb{E}}_{\boldsymbol{\lambda}}[h(\boldsymbol{\lambda})], C_r\}}{C_s-v_s}\right\rfloor,$$ where $\lfloor x\rfloor$ is the greatest integer less than or equal to $x$. Also, we know that the lower bounds of $n_B,r_B,m_B$ are 0. Therefore, we can say that we get a finite value of $\boldsymbol{q}_B=(n_B,r_{B},m_B,\tau_{1B})$. It is noted that for $m=0$, $\tau_{1}=0$. We now provide Algorithm A to obtain optimal AABSP.\\
\textbf{Algorithm A}
\begin{enumerate}
    \item For each parameter $\boldsymbol{q}=(n, r, m,\tau_1)$, derive the Bayes decision function $a_B(. \ | \ \boldsymbol{q})$ to minimize $R_B(\boldsymbol{q})$ among all classes of decision functions $a(. \ | \ \boldsymbol{q})$.  The derivation is discussed in Section \ref{Bayes decision}.
    \item When $m>0$, for each pair of values of $(n,r,m)$, minimize the Bayes risk with respect to $\tau_1$. Let $\tau_{1B}(n,r(n),m(r(n))$ be the optimal values of $\tau_1$.
    \begin{align*}
       &R_B(n,r(n),m(r(n)),\tau_{1B}(n,r(n),m(r(n))),a_B(. \ | \ (n,r(n),m(r(n)),\tau_{1B}(n,r(n),m(r(n))))))\\
       =&\min_{\tau_1\geq 0} R_B(n,r(n),m(r(n)),\tau_1(n,r(n),m(r(n))),a_B(. \ | \ (n,r(n),m(r(n)),,\tau_1(n,r(n),m(r(n)))))) 
    \end{align*}
    When $m=0$, we take $\tau_{1B}(n,r(n),m(r(n))=0$.
    \item For each pair of value $(n,r)$, find find an integer $m_B(n,r(n))$ which minimizes the equation (\ref{eq:bayesrisk}) with respect to $m$
  \small  \begin{align*}
       & R_B(n,r(n),m_B(r(n)),\tau_{1B}(n,r(n),m_B(r(n))),a_B(. \ | \ (n,r(n),m_B(r(n)),\tau_{1B}(n,r(n),m_B(r(n))))))\\
        =&\min_{0\leq m\leq r} R_B(n,r(n),m(r(n)),\tau_{1B}(n,r(n),m(r(n))),a_B(. \ | \ (n,r(n),m(r(n)),\tau_{1B}(n,r(n),m(r(n))))))
    \end{align*}
 \normalsize   \item For each value of $n$, find an integer $r_B(n)$ which minimizes the equation (\ref{eq:bayesrisk}) with respect to $r$
   \small \begin{align*}
       & R_B(n,r_B(n),m_B(r_B),\tau_{1B}(n,r_B(n),m_B(r_B(n))),a_B(. \ | \ (n,r_B(n),m_B(r_B(n)),\tau_{1B}(n,r_B(n),m(r_B(n))))))\\
        =&\min_{1\leq r\leq n} R_B(n,r(n),m_B(r(n)),\tau_{1B}(n,r(n),m_B(r(n))),a_B(. \ | \ (n,r_B(n),m_B(r(n)),\tau_{1B}(n,r(n),m(r(n))))))
    \end{align*}
  \normalsize  \item Finally, we find $n_B$ which minimize the equation (\ref{eq:bayesrisk})
  \begin{align*}
       & R_B(n_B,r_B(n_B),m_B(r_B(n_B)),\tau_{1B}(n_B,r_B(n_B),m_B(r_B(n_B))),\\
       &~~~~~~a_B(. \ | \ (n_B,r_B(n_B),m_B(r_B(n_B)),\tau_{1B}(n_B,r_B(n_B),m(r_B(n_B))))))\\
        =&\min_{0<n\leq n_0} R_B(n,r_B(n),m_B(r_B(n)),t_{1B}(n,r_B(n),m_B(r_B(n))),\\
        &~~~~~~~~~~~~~~~~~a_B(. \ | \ (n,t_{1B}(n,r_B(n),m(r_B(n))),r_B(n),m_B(r_B(n)))))
    \end{align*}
    \end{enumerate}
\normalsize Now, we take $a_B(. \ | \ (n_B,r_B(n_B),m_B(r_B(n_B)),\tau_{1B}(n_B,r_B(n_B),m(r_B(n_B))))=a_B$, $m_B(r_B(n_B))=m_B$, $r_B(n_B)=r_B$ and $\tau_{1B}(n_B,r_B(n_B),m_B(r_B(n_B)))=\tau_{1B}$.  Therefore, $(\boldsymbol{q}_B,a_B)=(n_B,r_B,m_{B},\tau_{1B},a_B)$ is the optimal sampling plan. 
    \begin{theorem}
        The sampling plan $(n_B,r_B,m_B,\tau_{1B},a_B)$ is the optimal AABSP
    \end{theorem}
    \begin{proof}
        It is enough to prove that for any sampling plan $(n,r,m,\tau_1,a)$, the following inequality holds:
        \begin{align*}
           [R_B(n,r,m,\tau_1,a)-R_B(n_B,r_B,m_B,\tau_{1B},a_B)]\geq 0.
        \end{align*}
        Now, the above difference can be written as
     \small   \begin{align*}
            R_B(n,r,m,\tau_1,a)-R_B(n_B,r_B,m_B,\tau_{1B},a_B)=&[R_B(n,r,m,\tau_1,a)-R_B(n,r,m,\tau_1,a_B)]\\
            &+[R_B(n,r,m,\tau_1,a_B)-R_B(n,r,m,\tau_{1B},a_B)]\\
            &+[R_B(n,r,m,\tau_{1B},a_B)-R_B(n,r,m_B,\tau_{1B},a_B)]\\
            &+[R_B(n,r,m_B,\tau_{1B},a_B)-R_B(n_B,r_B,m_B,\tau_{1B},a_B)]\\
            &+[R_B(n,r_B,m_B,\tau_{1B},a_B)-R_B(n_B,r_B,m_B,\tau_{1B},a_B)].
        \end{align*}
\normalsize From Algorithm A, it follows that   
  $[R_B(n,r,m,\tau_1,a)-R_B(n,r,m,\tau_1,a_B)]\geq 0$, $[R_B(n,r,m,\tau_1,a_B)-$\\ $R_B(n,r,m,\tau_{1B},a_B)]\geq 0$, $[R_B(n,r,m,\tau_{1B},a_B)-R_B(n,r,m_B,\tau_{1B},a_B)]\geq 0$, $[R_B(n,r,m_B,\tau_{1B},a_B)-$\\$R_B(n_B,r_B,m_B,\tau_{1B},a_B)]\geq0$,
  $[R_B(n,r_B,m_B,\tau_{1B},a_B)-R_B(n_B,r_B,m_B,\tau_{1B},a_B)]\geq 0$.
  Therefore, we get $[R_B(n,r,m,\tau_{1},a)-R_B(n_B,r_B,m_B,\tau_{1B},a_B)]\geq 0$. Hence, the proof is completed.
    \end{proof}\\
\section{Numerical study using real-life data and comparisons with other sampling plans}\label{num}
Here we illustrate the proposed method using an example and compute the optimum AABSP. Also, we study the effect of time cost components and hyperparameters on the optimum solution of AABSP. The optimum solution is compared with a CBSP and ACBSP. Here, we consider only two competing risks.

\subsection{An illustrative Example with comparisons}\label{example}
\begin{example}
    The hyperparameters of the prior are taken as $\alpha_1=2.5$, $\beta_1=1.5$, $\alpha_2=2.2$, $\beta_2=2.0$, $l_1=10$ and $l_2=10$. The cost coefficients of the loss function are taken as $a_0=2$, $a_1=3$, $a_2=3$, $a_{11}=4$, $a_{12}=4$ and $a_{22}=4$. The values of other cost components are $C_a=0.1$, $v_s=0.25$, $C_s=0.5$, $C_t=5$ and $C_r=40$.
\end{example}
In our proposed model, if we fix $m=0$, then the model becomes conventional non-accelerated BSP (CBSP) of \citet{prajapati2023optimal} and for $m=r$, the model becomes fully accelerated BSP (ACBSP) under SSSPALT under type-II censoring. For comparing the CBSP and ACBSP with AABSP, we calculate the percentage of relative risk savings (RRS) of a AABSP over CBSP and ACBSP, which are measured by $RRS_1$ and $RRS_2$, respectively, and provided by
\begin{align*}
    RRS_1=100\times \frac{R_1-R_{B}}{R_1}\%
\end{align*}
and 
\begin{align*}
    RRS_2=100\times \frac{R_2-R_{B}}{R_2}\%,
\end{align*}
where $R_1$ and $R_2$ are the optimal Bayes risk of CBSP and ACBSP, respectively. The optimal sampling plan $(n^*,r^*)$ and the optimal Bayes risk $R_1$ corresponding to CBSP are provided in Table \ref{opti}. The optimal sampling plan $(n_A,r_A,\tau_{1A})$ and the optimal Bayes risk $R_2$ of ACBSP, and for comparison $RRS_1$ and $RRS_2$ are given in Table \ref{opti}.
\normalsize\begin{table}[hbt!]
    \centering
      \caption{Optimal design}
    \begin{tabular}{|c|cc|cc|cc|cc|}
    \hline
  &  \multicolumn{2}{|c}{AABSP}&\multicolumn{2}{|c}{ACBSP}&\multicolumn{2}{|c}{CBSP}&\multicolumn{2}{|c|}{$RRS(\%)$}\\
    \hline
  $c_a$&$(n_B,r_B,m_B,\tau_{1B})$&$R_{B}$&$(n_A,r_A,\tau_{1A})$& $R_1$&$(n^*,r^*)$& $R_2$& $RRS_1$&$RRS_2$\\
 \hline
    0   & (5, 3, 3, 0.115)   & 35.964 & (5, 3, 0.115)   & 35.964 & (6, 3) & 36.582 & 0.000 & 1.718 \\
    0.1 & (5, 3, 2, 0.138)   & 36.252 & (4, 3, 0.142)   & 36.256 & (6, 3) & 36.582 & 0.011 & 0.893 \\
    0.2 & (5, 3, 2, 0.212)   & 36.477 & (5, 3, 0.301)   & 36.520 & (6, 3) & 36.582 & 0.170 & 0.289 \\
    \hline
    \end{tabular}
    \label{opti}
\end{table}

\subsection{Effect of the parameters}\label{effect}
In this section, we investigate the impact of cost components and hyperparameters on the optimal solution of the AABSP, in comparison with the CBSP and the ACBSP. The optimal designs corresponding to various values of the test cost $c_t$ and the hyperparameters $\alpha_1$ and $l_1$, under three different values of the acceleration cost coefficient $c_a = 0, 0.1, 0.2$, are summarized in Table~\ref{ct1}, Table~\ref{alpha1} and Table~\ref{l1}, respectively.
 Also, for the comparison of AABSP, CBSP and ACBSP, the $RRS_1$ and $RRS_2$ are tabulated. The other values of the cost components and hyperparameters, which are not mentioned in the tables, are kept fixed. 
\begin{table}[hbt!]
    \centering
      \caption{Optimal design for different $c_t$}
    \begin{tabular}{|cc|cc|cc|cc|cc|}
    \hline
   && \multicolumn{2}{|c}{AABSP}&\multicolumn{2}{|c}{ACBSP}&\multicolumn{2}{|c}{CBSP}&\multicolumn{2}{|c|}{$RRS(\%)$}\\
    \cline{3-10}
  $c_t$&$c_a$ & $(n_B,r_B,m_B,\tau_{1B})$&$R_{B}$&$(n_A,r_A,t_{1A})$& $R_1$&$(n^*,r^*)$& $R_2$& $RRS_1$&$RRS_2$\\
  \hline
   1&0& (4, 4, 4,  0.390)& 34.855&(4, 4,  0.390)& 34.855&(5, 4)&35.095&0.000&0.684\\
  1&0.1& (4, 4,  4,  0.539)& 34.986&  (4, 4,  0.539)& 34.986&(5, 4)& 35.095&0.000&0.351\\
  1&0.2&(4, 4, 4,  0.723)& 35.077& (4, 4,  0.723)& 35.077&(5, 4)& 35.095&0.000&0.051\\
\hline
2 & 0   & (5, 4, 3, 0.233)   & 35.250 & (4, 3, 0.240)   & 35.293 & (5, 3) & 35.633 & 0.122 & 1.077 \\
2 & 0.1 & (5, 4, 3, 0.282)   & 35.441 & (5, 4, 0.363) & 35.457 & (5, 3) & 35.633 & 0.045 & 0.540 \\
2 & 0.2 & (4, 3, 2, 0.296)   & 35.581 & (4, 3, 0.453)   & 35.593 & (5, 3) & 35.633 & 0.034 & 0.147 \\
\hline
3 & 0   & (5, 4, 4, 0.198)   & 35.547 & (5, 4, 0.198)   & 35.547 & (5, 3) & 35.995 & 0.000 & 1.245 \\
3 & 0.1 & (4, 3, 3, 0.247)   & 35.763 & (4, 3, 0.247)   & 35.763 & (5, 3) & 35.995 & 0.000 & 0.644 \\
3 & 0.2 & (4, 3, 2, 0.236)   & 35.918 & (4, 3, 0.315)   & 35.932 & (5, 3) & 35.995 & 0.039 & 0.214 \\
\hline
7 & 0   & (5, 2, 2, 0.152)   & 36.599 & (5, 2, 0.152)   & 36.599 & (5, 2) & 37.075 & 0.000 & 1.282 \\
7 & 0.1 & (4, 2, 2, 0.202)   & 36.795 & (4, 2, 0.202)   & 36.795 & (5, 2) & 37.075 & 0.000 & 0.755 \\
7 & 0.2 & (4, 2, 2, 0.211)   & 36.949 & (4, 2, 0.211) & 36.949 & (5, 2) & 37.075 & 0.000 & 0.339 \\
\hline
10&0&(5,  3,  3,  0.034)& 36.619&(5,  3,  0.034)& 36.619&(5, 2)&37.698&0.000&2.859\\
10&0.1&(3,  2,  2, 0.005) &36.942&(3,  2,  0.005) &36.942&(5, 2)&37.698&0.000&2.003\\
10&0.2&(3, 2, 2,  0.017)&37.236&(3, 2,  0.017)&37.236&(5, 2)&37.698&0.000&1.227\\
\hline
    \end{tabular}
    \label{ct1}
\end{table}

In Table~\ref{ct1}, it is observed that as the test cost $c_t$ increases, the corresponding time point $\tau_{1B}$ in the AABSP decreases. This trend arises because a higher time cost necessitates reducing the duration of the life test to minimize the overall expected cost. Consequently, the test is conducted for a shorter time under the lower stress level, leading to a decrease in the optimal time point $\tau_{1B}$ with increasing $c_t$.
It is seen that $RRS_2$ increases with $c_t$ and decreases with $c_a$, as expected. Also, in Table \ref{opti}, it is seen that when $c_t=5$ and $c_a=0.2$, AABSP is much better than the other two sampling plans. Also, it can be said that when $c_t$ is low and $c_a$ is high, CBSP is better than CABSP.

\begin{table}[hbt!]
    \centering
      \caption{Optimal design for different $\alpha_1$}
    \begin{tabular}{|cc|cc|cc|cc|cc|}
    \hline
   && \multicolumn{2}{|c}{AABSP}&\multicolumn{2}{|c}{ACBSP}&\multicolumn{2}{|c}{CBSP}&\multicolumn{2}{|c|}{$RRS(\%)$}\\
    \cline{3-10}
  $\alpha_1$&$c_a$ & $(n_B,r_B,m_B,\tau_{1B})$&$R_{B}$&$(n_A,r_A,t_{1A})$& $R_1$&$(n^*,r^*)$& $R_2$& $RRS_1$&$RRS_2$\\
  \hline
2&0& (4, 3,  3, 0.051)& 32.350&(4, 3,  0.051)& 32.350&(0, 0)& 32.873& 0.000 & 1.615 \\
  2&0.1& (4,  3, 3,  0.079)& 32.691&   (4,  3,  0.079)& 32.691&(0, 0)& 32.873& 0.000 & 0.557 \\
  2&0.2&(0, 0, 0, 0)& 32.873& (0, 0,  0)& 32.873&(0, 0)& 32.873& 0.000 & 0.000 \\
\hline
2.2 & 0   & (4, 3, 3, 0.074)& 33.907 & (4, 3, 0.074)& 33.907  & (5, 2) & 34.764 & 0.000 & 2.526 \\
2.2 & 0.1 &(4, 3, 3, 0.107)& 34.224 & (4, 3, 0.107)& 34.224  & (5, 2) & 34.764 & 0.000 & 1.577 \\
2.2 & 0.2 & (4, 3, 3,  0.149)& 34.509& (4, 3,  0.149)& 34.509 & (5, 2) & 34.764 & 0.000 & 0.739 \\
\hline
2.8 & 0   &(5, 3, 3, 0.146)& 37.678 & (5, 3,  0.146)& 37.678 & (6, 3)& 38.147&0.000&1.246  \\
2.8 & 0.1 & (5, 3,  2, 0.166)& 37.916 &  (5, 3,  0.215)& 37.928 & (6, 3)& 38.147&0.032&0.610  \\
2.8 & 0.2 & (5, 3, 3,  0.284)& 38.109& (5, 3,  0.284)& 38.109 &(6, 3)& 38.147&0.000&0.100  \\
\hline
3 & 0   & (5, 3, 3, 0.161)&38.638&  (5, 3, 0.161)&38.638&(5, 2)& 39.015&0.000&0.976  \\
3 & 0.1 &  (5, 3, 2, 0.184)& 38.844& (4, 2,  0.210)& 38.853 & (5, 2)& 39.015&0.023&0.440 \\
3 & 0.2 & (4, 2, 1, 0.170)& 38.955 & (4, 2,  0.268)& 38.967 & (5, 2)& 39.015&0.031&0.155 \\
\hline
    \end{tabular}
    \label{alpha1}
\end{table}

In Table~\ref{alpha1}, it is observed that the Bayes risk increases with $\alpha_1$, as expected. When $\alpha_1 = 2$ and $c_a = 0$, the lot is accepted without any life testing under the CBSP. This is because the Bayes risk simplifies to 
\[
R_B = a_0 + a_1 \frac{\alpha_1}{\beta_1} + a_2 \frac{\alpha_2}{\beta_2} + a_{11} \frac{\alpha_1(\alpha_1+1)}{\beta_1^2} + a_{12} \left( \frac{\alpha_1}{\beta_1} \cdot \frac{\alpha_2}{\beta_2} \right) + a_{22} \frac{\alpha_2(\alpha_2+1)}{\beta_2^2},
\]
which is minimal without life testing. However, for the AABSP, life testing is still conducted.
It is also noted that the relative risk saving $RRS_2$ generally increases with $\alpha_1$, indicating improved performance of the AABSP compared to the CBSP. Nevertheless, beyond a certain threshold of $\alpha_1$, $RRS_2$ starts to decrease. This is because the difference between the rejection cost and the prior expected acceptance cost initially decreases and then increases with $\alpha_1$.

\begin{table}[hbt!]
    \centering
      \caption{Optimal design for different $l_1$}
    \begin{tabular}{|cc|cc|cc|cc|cc|}
    \hline
   && \multicolumn{2}{|c}{AABSP}&\multicolumn{2}{|c}{ACBSP}&\multicolumn{2}{|c}{CBSP}&\multicolumn{2}{|c|}{$RRS(\%)$}\\
    \cline{3-10}
  $l_1$&$c_a$ & $(n_B,r_B,m_B,\tau_{1B})$&$R_{B}$&$(n_A,r_A,t_{1A})$& $R_1$&$(n^*,r^*)$& $R_2$& $RRS_1$&$RRS_2$\\
  \hline
\hline
    20  & 0   & (4, 3, 3, 0.113) & 35.904 & (4, 3, 0.113) & 35.904 & (6, 3) & 36.582 & 0.000 & 1.888 \\
    20  & 0.1 & (4, 3, 3, 0.170) & 36.011 & (4, 3, 0.170) & 36.011 & (6, 3) & 36.582 & 0.000 & 1.112 \\
    20  & 0.2 & (5, 3, 2, 0.169) & 36.403 & (4, 3, 3, 0.198) & 36.416 & (6, 3) & 36.582 & 0.036 & 0.492 \\
    \hline
    50  & 0   & (4, 3, 3, 0.132) & 35.866 & (4, 3, 3, 0.132) & 35.866 & (6, 3) & 36.582 & 0.000 & 1.996 \\
    50  & 0.1 & (4, 3, 3, 0.170) & 36.123 & (4, 3, 0.170) & 36.123 & (6, 3) & 36.582 & 0.000 & 1.271 \\
    50  & 0.2 & (4, 3, 2, 0.162) & 36.343 & (4, 3, 0.214) & 36.348 & (6, 3) & 36.582 & 0.014 & 0.658 \\
    \hline
    100 & 0   & (4, 3, 3, 0.146) & 35.856 & (4, 3, 0.146) & 35.856 & (6, 3) & 36.582 & 0.000 & 2.025 \\
    100 & 0.1 & (4, 3, 3, 0.187) & 36.102 & (4, 3, 0.187) & 36.102 & (6, 3) & 36.582 & 0.000 & 1.330 \\
    100 & 0.2 & (4, 3, 2, 0.167) & 36.305 & (4, 3, 0.228) & 36.317 & (6, 3) & 36.582 & 0.033 & 0.763 \\
    \hline
    500 & 0   & (5, 3, 2, 0.125) & 35.840 & (5, 3, 0.175) & 35.844 & (6, 3) & 36.582 & 0.011 & 2.070 \\
    500 & 0.1 & (4, 3, 2, 0.154) & 36.065 & (4, 3, 0.209) & 36.077 & (6, 3) & 36.582 & 0.033 & 1.434 \\
    500 & 0.2 & (4, 3, 2, 0.170) & 36.257 & (4, 3, 0.241) & 36.279 & (6, 3) & 36.582 & 0.061 & 0.896 \\
    \hline
    \end{tabular}
    \label{l1}
\end{table}

In Table \ref{l1}, it is seen that Bayes risk of AABSP decreases. This is due to fact that when $l$ increase the expected time duration can decreases. Also, the solution of CBSP is unchanged with $l_1$, as expected. In Tables \ref{opti}-\ref{l1}, the solution of CBSP is unchanged with $C_a$, as expected. 
\section{Real Data analysis}\label{com}
In the previous section, the time point is taken as a variable quantity that is determined by minimizing the Bayes risk. However, in some situations, consumers want life testing in normal operating conditions up to a certain time. If the consumer does not get sufficient failure time information, then they agree for higher stress. Therefore, in this scenario, the time point is pre-fixed. The other components of the designing parameter of $\boldsymbol{q}$ is determined by minimizing the Bayes risk.
\subsection{Data Analysis}\label{analysis} A solar lighting device sample of size $35$ was put on a life test at a normal operating temperature $293K$. After a time point $\tau_1=5$ (in hundred hours), the temperature was changed from $293K$ to $353K$. The life test is terminated after a time point $\tau_2=6$ (in hundred hours). The solar device has two failure modes: one is capacitor failure, denoted by 1, and another is controller failure, denoted by 2. The data is given in Table \ref{data}.
\begin{table}[hbt!]
    \centering
    \caption{A competing risk data under SSSPALT}
    \begin{tabular}{|c|c|l|}
    \hline
    Operating Conditions & Cause& Failure time \\
    \hline
    \multirow{3}{*}{ At $293K$ (before $\tau_1=5)$}     &1& 0.140, 1.582, 4.612  \\
    \cline{2-3}
      &\multirow{2}{*}{2}& 0.783, 1.324, 1.716, 1.794, 1.883, 2.293, 2.660, 2.674, 2.725,\\
        &&3.085, 3.924, 4.396, 4.892  \\
        \hline
  \multirow{2}{*}{  At $353K$ (after $\tau_1=5)$ }    & 1& 5.002, 5.112, 5.147, 5.238, 5.244, 5.247, 5.305, 5.407, 5.445, 5.483 \\
  \cline{2-3}
    & 2& 5.022, 5.082, 5.337, 5.408, 5.717\\
    \hline
    \end{tabular}
    \label{data}
\end{table}
For the data given in Table \ref{data}, we get the MLEs using the \citet{han2010inference}:  $\widehat\lambda_{1}=0.0222$, $\widehat\lambda_{2}=0.0960$, $\widehat\phi_1=55.356$ and $\widehat\phi_2=6.359$. The standard deviations of $\lambda_1$ and $\lambda_2$ are taken as 0.0127 and 0.0266, respectively, using the observed Fisher information matrix.
Now, using these MLEs, the reliability function for component $i$ under SSSALT can be written as 
\begin{align*}
    \overline{F}_i(t)=\begin{dcases*}
        \exp[-\widehat\lambda_{i}t]& if $t\leq \tau_1$\\
        \exp[-\widehat\lambda_{i}\left\{\tau_1+\widehat\phi_i(t-\tau_1)\right\}] &if $t>\tau_1$,
    \end{dcases*}
\end{align*}
for $i=1,2$ and  the reliability function for unit under SSSALT can be written as 
\begin{align*}
    \overline{F}(t)=\begin{dcases*}
        \exp[-(\widehat\lambda_{1}+\widehat\lambda_{2)}t]& if $t\leq \tau_1$\\
        \exp[-(\widehat\lambda_{1}+\widehat\lambda_{2})\tau_1-(\widehat\lambda_{1}\widehat\phi_1+\widehat\lambda_{2}\widehat\phi_2)(t-\tau_1)] &if $t>\tau_1$,
    \end{dcases*}
\end{align*}
Using these functions, we draw the parametric reliability curve for component 1, component 2, and for the unit, which are represented by dotted lines in Figure \ref{fig}. Also, using the Kaplan-Meier estimator, we draw the non-parametric reliability curve for component 1, component 2, and for the unit, which are represented by straight lines in Figure \ref{fig}. 
\begin{figure}[hbt!]
    \centering
    \includegraphics[width=0.33\linewidth]{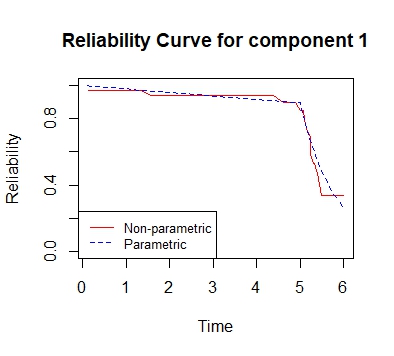}\includegraphics[width=0.33\linewidth]{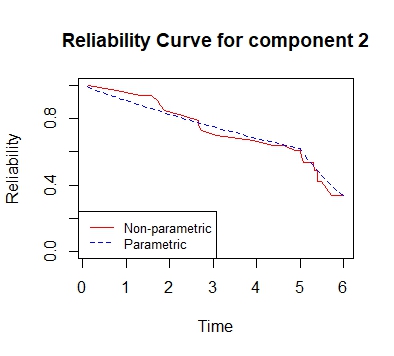}\includegraphics[width=0.33\linewidth]{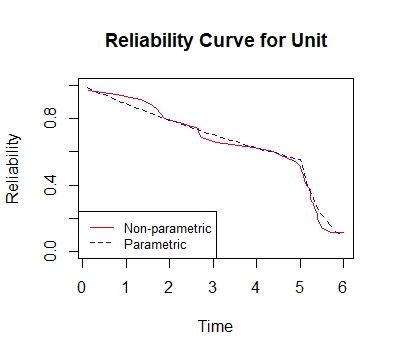}
    \caption{Non-parametric and parametric reliability curve}
    \label{fig}
\end{figure}

In Figure \ref{fig}, it is seen that the parametric curve and the non-parametric curve are close to each other for all cases. Therefore, we can say that lifetimes of components 1 and 2 are well-fitted in the exponential distribution and also, the lifetime of the unit is well-fitted in our assumed distribution. Therefore, we can say that in the data, components 1 and 2 are independent in nature. 
\subsection{Illustrative Example }

\begin{example}
    The hyperparameters of the priors are taken as $\alpha_{1}=3.05$, $\alpha_{2}=13$, $\beta_{1}=137.35$, $\beta_{2}=135.44$, $l_1=109.072$ and $l_2=11.718$ The values of other cost components are $c_t=0.3$, $c_s=0.5$, $v_s=0.3$, $c_a=0.05$, $c_r=80$. The cost components of the loss function are taken $a_{0}=6$, $a_{1}=200$, $a_{2}=200$, $a_{11}=4000$, $a_{22}=4000$, $a_{12}=4000$. the time point $\tau_1$ is taken as $5$.\\
\end{example} 
\begin{table}[hbt!]
    \centering
      \caption{Optimal design for different scenarios when $\tau_1=5$}
    \begin{tabular}{|cc|cc|cc|cc|}
    \hline
   \multicolumn{2}{|c}{AABSP}&\multicolumn{2}{|c}{ACBSP}&\multicolumn{2}{|c}{CBSP}&\multicolumn{2}{|c|}{RRS}\\
    \hline
  $(n_B,r_B,m_B)$&$R_{B}$&$(n_A,r_A)$& $R_1$&$(n^*,r^*)$& $R_2$& $RRS_1(\%)$&$RRS_2(\%)$\\
    \hline
 (7, 5, 3)&77.126&(6, 5)&77.283&(7, 5)&77.236&0.204	&0.142\\
      \hline
    \end{tabular}
    \label{optimal}
\end{table}

In the Table \ref{optimal}, AABSP is much better than other two sampling plans.
\subsection{Effect of the parameters}
\begin{table}[hbt!]
    \centering
      \caption{Optimal design for different scenarios for different $c_a$}
    \begin{tabular}{|c|cc|cc|cc|cc|}
    \hline
  &  \multicolumn{2}{|c}{AABSP}&\multicolumn{2}{|c}{ACBSP}&\multicolumn{2}{|c}{CBSP}&\multicolumn{2}{|c|}{$RRS(\%)$}\\
    \hline
  $c_a$&     $(n_B,r_B,m_B)$&$R_{B}$&$(n_A,r_A)$& $R_1$&$(n^*,r^*)$& $R_2$& $RRS_1$&$RRS_2$\\
    \hline
   0& (7, 5, 3)&77.030&(6, 5)&77.119&(6, 4)&77.236&0.115&0.267\\
     0.1& (7, 5, 2)&77.196&(6, 5)&77.443&(7, 5)&77.236&0.319&0.052\\
       0.15& (7, 5, 1)&77.234&(5, 5)&77.583&(7, 5)&77.236&0.449&0.003\\
       0.2& (7, 5, 0)&77.236&(4, 4)&77.708&(7, 5)&77.236&0.607&0.000\\
      \hline
    \end{tabular}
  
    \label{ca}
\end{table}

In Table \ref{ca}, it is seen that when $c_a$ increases the value of $m$ decreases. Also, it is seen that $RRS_1$ increases with $c_a$ and $RRS_2$ decreases with $c_a$. This indicates that AABSP converges to CBSP as expected. 
\begin{table}[hbt!]
    \centering
    \caption{Optimal design for different scenarios for different $c_t$ (rounded to 3 decimal places)}
    \begin{tabular}{|c|cc|cc|cc|cc|}
    \hline
     &  \multicolumn{2}{|c}{AABSP}&\multicolumn{2}{|c}{ACBSP}&\multicolumn{2}{|c}{CBSP}&\multicolumn{2}{|c|}{$RRS(\%)$}\\
    \hline
    $c_t$ 
    & $(n_B,r_B,m_B)$ & $R_{B}$ 
    & $(n_A,r_A)$ & $R_1$ 
    & $(n^*,r^*)$ & $R_2$ 
    & $RRS_1$ & $RRS_2$ \\
    \hline
    0    & (5, 5, 0)   & 75.850 & (5, 5)   & 76.664 & (5, 5)   & 75.850 & 1.061 & 0.000 \\
    0.05 & (6, 5, 2)   & 76.661 & (5, 5)   & 76.983 & (6, 5)   & 76.704 & 0.418 & 0.056 \\
    0.15 & (7, 5, 3)   & 77.481 & (6, 5)   & 77.566 & (7, 4)   & 77.685 & 0.110 & 0.263 \\
    0.2  & (7, 4, 3)   & 77.806 & (7, 5)   & 77.837 & (8, 4)   & 78.002 & 0.040 & 0.251 \\
    0.3  & (7, 3, 3)   & 78.309 & (7, 3)   & 78.309 & (8, 3)   & 78.556 & 0.000 & 0.314 \\
    \hline
    \end{tabular}
    \label{ct}
\end{table}

In Table \ref{ct}, it is seen that for higher values of $c_t$, the AABSP is equivalent to ACBSP and for $c_t=0$, AABSP is equivalent to CBSP. The fact is already discussed in Table \ref{ct1}.
\begin{table}[hbt!]
    \centering
    \caption{Optimal design for different scenarios for different $c_r$ (rounded to 3 decimal places)}
    \begin{tabular}{|c|cc|cc|cc|cc|}
    \hline
     &  \multicolumn{2}{|c}{AABSP}&\multicolumn{2}{|c}{ACBSP}&\multicolumn{2}{|c}{CBSP}&\multicolumn{2}{|c|}{$RRS(\%)$}\\
    \hline
    $c_r$ 
    & $(n_B,r_B,m_B)$ & $R_{B}$ 
    & $(n_A,r_A)$ & $R_1$ 
    & $(n^*,r^*)$ & $R_2$ 
    & $RRS_1$ & $RRS_2$ \\
    \hline
    70 & (0, 0, 0)   & 70.000 & (0, 0)   & 70.000 & (0, 0)   & 70.000 & 0.000 & 0.000 \\
    75 & (7, 5, 2)   & 74.219 & (5, 4)   & 74.407 & (6, 4)   & 74.275 & 0.253 & 0.075 \\
    85 & (6, 5, 3)   & 79.356 & (6, 5)   & 79.444 & (7, 5)   & 79.550 & 0.111 & 0.244 \\
    90 & (0, 0, 0)   & 80.469 & (0, 0)   & 80.469 & (0, 0)   & 80.469 & 0.000 & 0.000 \\
    \hline
    \end{tabular}
    \label{cr}
\end{table}

In Table \ref{cr}, we provide the effect of $C_r$. For $C_r=70,90$, the optimal BSPs represent no sampling cases. For $C_r=70$, the Bayes risk is $R_B=C_r$, that is, the lot is rejected without life testing. For $C_r=90$, $R_B=a_0+a_1\alpha_1/\beta_1+a_2\alpha_2/\beta_2+a_{11}\alpha_1(\alpha_1+1)/\beta_1^2+a_{12}(\alpha_1/\beta_1)(\alpha_2/\beta_2)+a_{22}\alpha_2(\alpha_2+1)/\beta_2^2$ and the lot is accepted without life testing. For $C_r=75,85$, the optimal sampling parameters have sampling cases. 
\begin{table}[hbt!]
    \centering
    \caption{Optimal design for different scenarios for different $C_s$ (rounded to 3 decimal places)}
    \begin{tabular}{|c|cc|cc|cc|cc|}
    \hline
     &  \multicolumn{2}{|c}{AABSP}&\multicolumn{2}{|c}{ACBSP}&\multicolumn{2}{|c}{CBSP}&\multicolumn{2}{|c|}{$RRS(\%)$}\\
    \hline
    $C_s$ 
    & $(n_B,r_B,m_B)$ & $R_{B}$ 
    & $(n_A,r_A)$ & $R_1$ 
    & $(n^*,r^*)$ & $R_2$ 
    & $RRS_1$ & $RRS_2$ \\
    \hline
    0.35 & (6, 4, 3) & 77.475 & (5, 5) & 77.553 & (7, 5) & 77.586 & 0.101 & 0.143 \\
    0.40 & (5, 4, 2) & 77.736 & (4, 4) & 77.772 & (6, 4) & 77.914 & 0.046 & 0.229 \\
    0.45 & (4, 4, 4) & 77.972 & (4, 4) & 77.972 & (4, 3) & 78.158 & 0.000 & 0.238 \\
    0.50 & (3, 3, 3) & 78.129 & (3, 3) & 78.129 & (4, 3) & 78.358 & 0.000 & 0.292 \\
    \hline
    \end{tabular}
    \label{Cs}
\end{table}
In table \ref{Cs}, sample size is decreases with $c_s$, as expected. $RRS_1$ is decreases with $c_a$ but $RRS_2$ is increases with $c_s$. This is due to fact that AABSP is coincide with ACBSP for high value of $c_s$.
\subsection{Decision for some simulated data}
 We generate Type-II adaptive step-stress competing risk data based on the life testing plan $(n,r,m,\tau_1)=(7, 5, 3, 5)$ from independent the exponential distribution with $\lambda_{11}=0.0221$, $\lambda_{12}=0.0959$ and $\phi_1=55.101$ and $\phi_2=6.358$. The data sets $(\boldsymbol{y},\boldsymbol{d})=(y_1,y_2,y_3,y_4,y_5,d_{11},d_{12},d_{21},d_{22})$ and the corresponding decisions about the lot acceptance or rejection $a_B$ are given in Table \ref{ttt} for illustration purposes. Also, in Table \ref{ttt}, the decision of changing the stress at $\tau_{1}$, $w_1$  and $w_2$ are provided. 

\begin{table}[hbt!] 
    \centering
 \caption{Simulated data sets and corresponding decision about the lot }
   \begin{tabular}{|c|ccccccccccccc|}
    \hline
       $i$& $y_1$ &$y_2$ &$y_3$ &$y_4$&$y_5$& $d_{11}$&$d_{12}$&$d_1$& Change & $d_{21}$&$d_{22}$& $e_i$& $a_B$ \\
        &  && &&&&&& the stress? &&&  &  \\
        \hline
     1& 2.77& 10.11&  5.42&  5.07&  8.08&1 &1&2&Yes& 2& 1&-25.69&1\\
     2&   6.25& 3.38& 5.37& 5.91& 3.92& 0 &4&4&No&0& 1&8.800&0\\
       3& 5.79 &4.87 &6.42 &5.21 &5.02& 1& 0&1&Yes& 0& 4&-7.85&1\\
4& 6.81  &3.70&  1.57& 18.96 & 5.73& 1&2&3& No&0& 2&-4.38&1\\
  5&    6.43& 3.17 &5.30& 4.30& 5.66& 1& 1&2&Yes& 0& 3&-6.72&1\\
 6&   0.98 &13.72&  4.07&  2.86& 18.37&  0&  4&4&No&  0&  1&-17.88&1\\
  7& 5.69&  0.30& 11.59&  2.13&  4.69&  1&  4&5&No&  0&  0&5.88&0\\
        \hline
    \end{tabular}
    \label{ttt}
\end{table}
The BSP can be illustrated as follows. Seven items are put on life test at $293K$. After time $\tau_1=5$, the number of failures is observed. If the number of failures is less than $3$, the stress is increased to $353K$ and the test continues up to $5$ failures observed. If the number of failures is greater than equals to $3$, the stress is unchanged and the test continues up to $5$ failures observed at $293K$.   Then we calculate  $e_i=\varphi(\boldsymbol{y}_i,\boldsymbol{d}_i)-C_r$, for $i=1,\ldots,7$ If  $e_i<0$, then $a_B=1$ and the lot is accepted. If $e_i\geq 0$ then $a_B=0$ and the lot is rejected. In data no. 7, it is seen that the fifth failure is observed before the time $\tau_1=5$. Therefore, the life test is terminated at 4.69 in a hundred hours.
\section{Conclusion}\label{con}
This work proposes an Adaptive BRASP based on a simple step-stress test for type-II competing censored data. Through numerical studies, it is observed that when the stress-changing time point is treated as a variable, the AABSP performs significantly better than the other two plans—CBSP and ACBSP-under certain cost configurations. In most cases, the Bayes risk of AABSP either closely matches or coincides with that of the ACBSP, highlighting its effectiveness. Notably, AABSP consistently outperforms CBSP, especially in scenarios involving high time cost ($c_t$) and low additional acceleration cost ($c_a$).

Although ACBSP has not been widely explored in existing literature, the current study emphasizes its practical value alongside AABSP. The observed Relative Risk Saving ($RRS$) from CBSP reaches up to approximately 2.5\%, suggesting substantial cost efficiency. Furthermore, when the stress-changing time point is fixed, AABSP demonstrates significantly better performance compared to both CBSP and ACBSP. These findings underscore the potential of AABSP as a robust and cost-effective methodology under censored data and competing risks in context of BRASP.

In this work, we have considered the exponential distribution for illustration. However, the proposed method can be extended to other lifetime distributions. The work can also be extended to other censoring schemes and more than 2-stage step-stress test.
\bibliography{citation}

\begin{thebibliography}{26}
\providecommand{\natexlab}[1]{#1}
\providecommand{\url}[1]{\texttt{#1}}
\expandafter\ifx\csname urlstyle\endcsname\relax
  \providecommand{\doi}[1]{doi: #1}\else
  \providecommand{\doi}{doi: \begingroup \urlstyle{rm}\Url}\fi

\bibitem[Balakrishnan and Han(2008)]{balakrishnan2008exact}
N.~Balakrishnan and D.~Han.
\newblock {Exact inference for a simple step-stress model with competing risks for failure from exponential distribution under Type-II censoring}.
\newblock \emph{Journal of Statistical Planning and Inference}, 138\penalty0 (12):\penalty0 4172--4186, 2008.

\bibitem[Balakrishnan et~al.(2009)Balakrishnan, Xie, and Kundu]{balakrishnan2009exact}
N.~Balakrishnan, Q.~Xie, and D.~Kundu.
\newblock {Exact inference for a simple step-stress model from the exponential distribution under time constraint}.
\newblock \emph{Annals of the Institute of Statistical Mathematics}, 61:\penalty0 251--274, 2009.

\bibitem[Chen et~al.(2007)Chen, Li, and Lam]{chen2007bayesian}
J.~Chen, K.-H. Li, and Y.~Lam.
\newblock {Bayesian single and double variable sampling plans for the Weibull distribution with censoring}.
\newblock \emph{European Journal of Operational Research}, 177\penalty0 (2):\penalty0 1062--1073, 2007.

\bibitem[Chen et~al.(2022{\natexlab{a}})Chen, Liang, and Yang]{chen2022designing}
L.-S. Chen, T.~Liang, and M.-C. Yang.
\newblock {Designing Bayesian sampling plans for simple step-stress of accelerated life test on censored data}.
\newblock \emph{Journal of Statistical Computation and Simulation}, 92\penalty0 (2):\penalty0 395--415, 2022{\natexlab{a}}.

\bibitem[Chen et~al.(2022{\natexlab{b}})Chen, Liang, and Yang]{chen2023designing}
L.-S. Chen, T.~Liang, and M.-C. Yang.
\newblock {Designing efficient Bayesian sampling plans based on the simple step-stress test under random stress-change time for censored data}.
\newblock \emph{Journal of Statistical Computation and Simulation}, pages 1--26, 2022{\natexlab{b}}.

\bibitem[Chen et~al.(2025)Chen, Liang, and Yang]{chen2025designing}
L.-S. Chen, T.~Liang, and M.-C. Yang.
\newblock {Designing efficient Bayesian sampling plans for exponential distributions based on samples under a (n, t) simple step-stress test}.
\newblock \emph{Communications in Statistics-Simulation and Computation}, 54\penalty0 (2):\penalty0 382--405, 2025.

\bibitem[Das and Pradhan(2024)]{das2024bayesian}
R.~Das and B.~Pradhan.
\newblock Bayesian reliability acceptance sampling plans under adaptive simple step stress partial accelerated life test.
\newblock \emph{arXiv preprint arXiv:2408.00734}, 2024.

\bibitem[Han and Balakrishnan(2010)]{han2010inference}
D.~Han and N.~Balakrishnan.
\newblock {Inference for a simple step-stress model with competing risks for failure from the exponential distribution under time constraint}.
\newblock \emph{Computational Statistics \& Data Analysis}, 54\penalty0 (9):\penalty0 2066--2081, 2010.

\bibitem[Han and Kundu(2014)]{han2014inference}
D.~Han and D.~Kundu.
\newblock {Inference for a step-stress model with competing risks for failure from the generalized exponential distribution under type-I censoring}.
\newblock \emph{IEEE Transactions on reliability}, 64\penalty0 (1):\penalty0 31--43, 2014.

\bibitem[Kalbfleisch and Prentice(2002)]{kalbfleisch2002statistical}
J.~D. Kalbfleisch and R.~L. Prentice.
\newblock \emph{{The Statistical Analysis of Failure Time Data}}.
\newblock John Wiley \& Sons, 2nd edition, 2002.

\bibitem[Kundu and Ganguly(2017)]{kundu2017analysis}
D.~Kundu and A.~Ganguly.
\newblock \emph{{Analysis of Step-Stress Models: Existing results and some recent developments}}.
\newblock Academic Press, 2017.

\bibitem[Lawless(2011)]{lawless2011statistical}
J.~F. Lawless.
\newblock \emph{{Statistical Models and Methods for Lifetime Data}}.
\newblock John Wiley \& Sons, 2nd edition, 2011.

\bibitem[Liang and Yang(2013)]{liang2013optimal}
T.~Liang and M.-C. Yang.
\newblock {Optimal Bayesian sampling plans for exponential distributions based on hybrid censored samples}.
\newblock \emph{Journal of Statistical Computation and Simulation}, 83\penalty0 (5):\penalty0 922--940, 2013.

\bibitem[Lin et~al.(2008)Lin, Huang, and Balakrishnan]{lin2008exact}
C.-T. Lin, Y.-L. Huang, and N.~Balakrishnan.
\newblock {Exact Bayesian variable sampling plans for the exponential distribution based on Type-I and Type-II hybrid censored samples}.
\newblock \emph{Communications in Statistics—Simulation and Computation{\textregistered}}, 37\penalty0 (6):\penalty0 1101--1116, 2008.

\bibitem[Lin et~al.(2010)Lin, Huang, and Balakrishnan*]{lin2010corrections}
C.-T. Lin, Y.-L. Huang, and N.~Balakrishnan*.
\newblock {Corrections on “Exact Bayesian variable sampling plans for the exponential distribution based on type-I and type-II hybrid censored samples”}.
\newblock \emph{Communications in Statistics-Simulation and Computation}, 39\penalty0 (7):\penalty0 1499--1505, 2010.

\bibitem[Lin et~al.(2002)Lin, Liang, and Huang]{lin2002bayesian}
Y.-P. Lin, T.~Liang, and W.-T. Huang.
\newblock {Bayesian sampling plans for exponential distribution based on type I censoring data}.
\newblock \emph{Annals of the Institute of Statistical Mathematics}, 54\penalty0 (1):\penalty0 100--113, 2002.

\bibitem[Nelson(2009)]{nelson2009accelerated}
W.~B. Nelson.
\newblock \emph{{Accelerated testing: statistical models, test plans, and data analysis}}.
\newblock John Wiley \& Sons, 2009.

\bibitem[Prajapati and Kundu(2024)]{prajapati2024bayesian}
D.~Prajapati and D.~Kundu.
\newblock {Bayesian acceptance sampling plan for simple step-stress model}.
\newblock \emph{Communications in Statistics-Simulation and Computation}, 53\penalty0 (12):\penalty0 6411--6431, 2024.

\bibitem[Prajapati et~al.(2019)Prajapati, Mitra, and Kundu]{prajapati2019new}
D.~Prajapati, S.~Mitra, and D.~Kundu.
\newblock {A new decision theoretic sampling plan for type-I and type-I hybrid censored samples from the exponential distribution}.
\newblock \emph{Sankhya B}, 81\penalty0 (2):\penalty0 251--288, 2019.

\bibitem[Prajapati et~al.(2023)Prajapati, Mitra, Kundu, and Pal]{prajapati2023optimal}
D.~Prajapati, S.~Mitra, D.~Kundu, and A.~Pal.
\newblock {Optimal Bayesian sampling plan for censored competing risks data}.
\newblock \emph{Journal of Statistical Computation and Simulation}, 93\penalty0 (5):\penalty0 775--799, 2023.

\bibitem[Samanta et~al.(2019)Samanta, Gupta, and Kundu]{samanta2019analysis}
D.~Samanta, A.~Gupta, and D.~Kundu.
\newblock {Analysis of Weibull step-stress model in presence of competing risk}.
\newblock \emph{IEEE Transactions on Reliability}, 68\penalty0 (2):\penalty0 420--438, 2019.

\bibitem[Sedyakin(1966)]{sedyakin1966relationship}
I.~V. Sedyakin.
\newblock {The Relationship Between Life Distributions of Populations Under Different Conditions}.
\newblock In \emph{Mathematical Methods in Reliability Theory}, pages 119--132, Moscow, 1966. Publishing House of the Academy of Sciences of the USSR.

\bibitem[Xiang et~al.(2017)Xiang, Yang, and Shen]{xiang2017designing}
S.~Xiang, J.~Yang, and L.~Shen.
\newblock {Designing adaptive accelerated life tests using Bayesian methods}.
\newblock \emph{Quality and Reliability Engineering International}, 33\penalty0 (8):\penalty0 2269--2279, 2017.

\bibitem[Yeh(1990)]{yeh1990optimal}
L.~Yeh.
\newblock {An optimal single variable sampling plan with censoring}.
\newblock \emph{Journal of the Royal Statistical Society: Series D (The Statistician)}, 39\penalty0 (1):\penalty0 53--66, 1990.

\bibitem[Yeh(1994)]{yeh1994bayesian}
L.~Yeh.
\newblock {Bayesian variable sampling plans for the exponential distribution with type I censoring}.
\newblock \emph{The Annals of Statistics}, pages 696--711, 1994.

\bibitem[Yeh and Choy(1995)]{yeh1995bayesian}
L.~Yeh and S.~Choy.
\newblock {Bayesian variable sampling plans for the exponential distribution with uniformly distributed random censoring}.
\newblock \emph{Journal of Statistical Planning and Inference}, 47\penalty0 (3):\penalty0 277--293, 1995.

\end{thebibliography}
\appendix

\section{Additional stress}\label{appa}
$n_{as}$ denotes the expected number of items put on higher stress levels after $\tau_1$.
\begin{align*}
    n_{as}&={\mathbb{E}}_{\boldsymbol{\theta}}[\mathbb{E}[(n-D_1)\ | \ D_1<m, \boldsymbol{\theta}]]\\
&=\int_{\lambda_1=0}^\infty\cdots\int_{\lambda_J=0}^\infty\sum_{d_1=0}^{m-1}(n-d_1)\binom{n}{d_1}(1-\exp(-\lambda \tau_1))^{d_1}\exp(-(n-d_1)\lambda \tau_1)~\left[\prod_{j=1}^Jp_{1j}(\lambda_j)d\lambda_j\right]\\ 
    &=\sum_{d_1=0}^{m-1}\sum_{i=0}^{d_1}(n-d_1)\binom{n}{d_1}\binom{d_1}{i}(-1)^{d_1-i}\prod_{j=1}^J\left[\int_{\lambda_j=0}^\infty\exp(-(n-i)\lambda_j \tau_1)~p_{1j}(\lambda_j)d\lambda_j\right],
\end{align*}
where $\lambda=\lambda_1+\cdots+\lambda_J$.
\section{ Expected time duration}\label{appb}
The time duration under type-II censoring is $T^{(r)}$.
For $0<t<\tau_1$
\allowdisplaybreaks{\begin{align*}
{\mathbb{E}}[T^{(r)} | \ \boldsymbol{\theta}]
=&\int_{t=0}^{\tau_1}R_{T^{(r)}\ | \ \boldsymbol{\theta}}(t)~dt+\int_{t=\tau_1}^{\infty}R_{T^{(r)}}(t\ | \ \boldsymbol{\theta})~ dt\\
\end{align*}}
For $0<t<\tau_1$
\begin{align*}
   R_{T^{(r)}}(t\ | \ \boldsymbol{\theta})&=P(\text{No. of failures in the interval $(0,t)$ is less than $r$})\\
   &=\sum_{d_1=0}^{r-1}\binom{n}{d_1}[1-R(t\ | \ \boldsymbol{\theta})]^{d_1}[R(t\ | \ \boldsymbol{\theta})]^{n-d_1}\\
   &=\sum_{d_1=0}^{r-1}\sum_{j=0}^{d_1}\binom{n}{d_1}\binom{d_1}{j}(-1)^{d_1-j}\exp[-\lambda(n-j)t]~~~~~~~~~~~~~~~~0<t<\tau_1
\end{align*}
Therefore,
\begin{align*}
\int_0^{\tau_1} R_{T^{(r)}}(t\ | \ \boldsymbol{\theta})~dt=\sum_{d_1=0}^{r-1}\sum_{j=0}^{d_1}\binom{n}{d_1}\binom{d_1}{j}(-1)^{d_1-j}\frac{1-\exp\left[-\left(\sum_{i=1}^J\lambda_i\right)(n-j)\tau_1\right]}{\left(\sum_{i=1}^J\lambda_i\right)(n-j)}
\end{align*}
For $t>\tau_1$
\small \begin{align*}
&~~~~~~R_{T^{(r)}}(t\ | \ \boldsymbol{\theta})\\&=P(\text{No. of failures in the interval $(0,t)$ is less than $r$ and  $D_1<r$}\ | \ \boldsymbol{\theta})\\
&=P(\text{No. of failures in the interval $(0,t)$ is less than $r$}\ |\  D_1< r,\boldsymbol{\theta})P(D_1<r\ | \ \boldsymbol{\theta})\\
&=\sum_{d_1=0}^{r-1}\sum_{d_2=0}^{r-d_1-1}\frac{n!}{d_1!d_2!(n-d)!}[1-R(\tau_1\ | \ \boldsymbol{\theta})]^{d_1}[R(\tau_1\ | \ \boldsymbol{\theta})-R(t\ | \ \boldsymbol{\theta})]^{d_2}[R(t\ | \ \boldsymbol{\theta})]^{n-d}\\
&=\sum_{d_1=0}^{r-1}\sum_{d_2=0}^{r-d_1-1}\frac{n!}{d_1!d_2!(n-d)!}[1-\exp(-\lambda \tau_1)]^{d_1}[\exp(-\lambda \tau_1)-\exp(-\lambda (\tau_1+\phi(t-\tau_1))]]^{d_2}[\exp(-\lambda (\tau_1+\phi(t-\tau_1))]^{n-d}\\
&=\sum_{d_1=0}^{r-1}\sum_{d_2=0}^{r-d_1-1}\sum_{j=0}^{d_1}\sum_{k=0}^{d_2}\frac{n!}{d_1!d_2!(n-d)!}(-1)^{j+k}\exp[-\lambda (n-d_1+j)\tau_1]\exp[-(\lambda_1\phi_1^\delta+\lambda_2\phi_2^\delta)(n-d+k) (t-\tau_1)].
\end{align*}
\normalsize Therefore,
\begin{align*}
&\int_{t=\tau_1}^\infty R_{T^{(r)}}(t\ | \ \boldsymbol{\theta}) ~dt\\
  =& \sum_{d_1=0}^{r-1}\sum_{d_2=0}^{r-d_1-1}\sum_{j=0}^{d_1}\sum_{k=0}^{d_2}\frac{n!}{d_1!d_2!(n-d)!}(-1)^{j+k}\exp[-\lambda (n-d_1+j)\tau_1]\frac{1}{(\lambda_1\phi_1^\delta+\lambda_2\phi_2^\delta)(n-d+k)} 
\end{align*}
\section{The expression of $f_{(W_1,W_2,\boldsymbol{D})\ | \ \boldsymbol{\theta})}(w_1,w_2,\boldsymbol{d}\ | \ \boldsymbol{\theta})$:}\label{appc}
The joint distribution of $\boldsymbol{D}$ is
\begin{align*}
    P(\boldsymbol{D}=\boldsymbol{d}\ | \ \boldsymbol{\theta})=A(\boldsymbol{d})\prod_{i=1}^2\prod_{j=1}^Jp_{ij}^{d_{ij}}[F(\tau_1\ | \ \boldsymbol{\theta})]^{d_1}[1-F(\tau_1\ | \ \boldsymbol{\theta})]^{n-d_1},
\end{align*}
where $A(\boldsymbol{d})=\binom{n}{d_1}\binom{d_1}{d_{11},d_{12},\ldots, d_{1J}}\binom{d_2}{d_{21},d_{22},\ldots, d_{2J}}$

\begin{theorem}(\citet{balakrishnan2009exact})
\footnotesize\begin{align*}
&f_{(W_1,W_2,\boldsymbol{d})\ | \ \boldsymbol{\theta})}(w_1,w_2,\boldsymbol{d}\ | \ \boldsymbol{\theta})\\
=&\begin{dcases*}  
h_{n\tau_1}(w_1) \binom{r}{d_{21},d_{22},\ldots, d_{2J}}\prod_{j=1}^Jp_{2j}^{d_{2j}}\gamma\left(w_2,r,{\phi_1\lambda_1+\phi_2\lambda_2}\right)&$d_1=0$\\
A(\boldsymbol{d})\sum_{j=0}^{d_1}(-1)^{j}\binom{d_1}{j}\prod_{i=1}^2\prod_{j=1}^Jp_{ij}^{d_{ij}}\exp[-\lambda(n-d_{1}+j)\tau_1]\gamma\left(w_1-(n-d_1+j)\tau_1),d_{1},{\lambda}\right)\gamma\left(w_2,d_{2},{\phi_1\lambda_1+\phi_2\lambda_2}\right)&$d_{1}>0$\\
\sum_{k=r}^n\binom{n}{k}\binom{r}{d_{11},\ldots,d_{1J}}\sum_{j'=0}^{k}(-1)^{j'}\binom{k}{j'}\prod_{j=1}^Jp_{1j}^{d_{1j}}\exp[-\lambda(n-k+j')\tau_1]\gamma(w_1-(n-k+j')\tau_1,r,\lambda)h_0(w_2)&$d_1=r$
\end{dcases*}
 \end{align*}
\normalsize    where the function $\gamma(x,\alpha,\beta)$ is the pdf of the gamma distribution with parameters $(\alpha,\beta)$ which is defined as
\begin{align*}  
\gamma(x,\alpha,\beta)=\begin{cases}
    \frac{\beta^{\alpha}}{\Gamma(\alpha)}x^{\alpha-1}\exp(-\beta x)& \text{for } x>0\\
    0 & \text{otherwise},
    \end{cases}
\end{align*}
the function $h_y(x)$ is a degenerate distribution at the point $y$, which is defined as \begin{align}
    h_y(x)=\begin{cases}
    1& \text{for } x=y\\
    0 & \text{otherwise},
    \end{cases}
\end{align}
\end{theorem}

\section{Expression of $R_1$}\label{appd}
The expression of $R_1(\boldsymbol{q})$ is given In (\ref{R1d}) as
\begin{align*}
    R_1(\boldsymbol{q})={\mathbb{E}}_{\boldsymbol{\lambda}}[h(\boldsymbol{\lambda}]+\sum H(\boldsymbol{d}),
\end{align*}
where $${\mathbb{E}}_{\boldsymbol{\lambda}}[h(\boldsymbol{\lambda})]=a_0+\sum_{j=1}^Ja_j\frac{\alpha_j}{\beta_j}+\underset{i< j}{\sum_{i=1}^J\sum_{j=1}^J}a_{ij}\frac{\alpha_i\alpha_j}{\beta_i\beta_j}+{\sum_{j=1}^J}a_{jj}\frac{\alpha_j(\alpha_j+1)}{\beta_j^2}.$$
For the expression of $H(\boldsymbol{d})$, we need to consider three cases:\\\\
\textbf{Case 1 : }When $d_1=0$,
\begin{align*}   H(\boldsymbol{0},\boldsymbol{d}_2)&=\binom{r}{d_{21},d_{22},\ldots, d_{2J}}\int_{w_2=0}^{c_1(n*\tau_1,\boldsymbol{d})}~I(n\tau_1\leq c'(\boldsymbol{d}))\frac{w_2^{d_2-1}}{\Gamma(d_2)}\Bigg[a_0H_{1}(n\tau_1,w_2,\boldsymbol{d},\boldsymbol{0}) \\
&+\sum_{j=0}^Ja_j~H_{1}(n\tau_1,w_2,\boldsymbol{d},\boldsymbol{p}_j)+\underset{i\leq j}{\sum_{i=1}^J\sum_{j=1}^J}a_{ij}~H_{1}(n\tau_1,w_2,\boldsymbol{d},\boldsymbol{p}_{ij})~\Bigg]~dw_2
\end{align*}
\textbf{Case 2 : }When $\boldsymbol{d}>\boldsymbol{0} $,
\begin{align*}
H(\boldsymbol{d})=&A(\boldsymbol{d})\sum_{j=0}^{d_1}(-1)^{j}\binom{d_1}{j}\int_{w_1=(n-d_1+j)\tau_1}^{c'(\boldsymbol{d})}\int_{w_2=0}^{c_1(w_1,\boldsymbol{d})}\frac{(w_1-(n-d_1+j)\tau_1)^{d_1-1}w_2^{d_2-1}}{\Gamma(d_1)\Gamma(d_2)}\Bigg[a_0H_{1}(w_1,w_2,\boldsymbol{d},\boldsymbol{0}) \\
&+\sum_{j=0}^Ja_j~H_{1}(w_1,w_2,\boldsymbol{d},\boldsymbol{p}_j)+\underset{i\leq j}{\sum_{i=1}^J\sum_{j=1}^J}a_{ij}~H_{1}(w_1,w_2,\boldsymbol{d},\boldsymbol{p}_{ij})~\Bigg]~dw_2~dw_1
\end{align*}
\textbf{Case 3 : }When $d_2=0$,
\begin{align*}
    H(\boldsymbol{d}_1,\boldsymbol{0})=&\sum_{k=r}^n\binom{n}{k}\binom{r}{d_{11},\ldots,d_{1J}}\sum_{j'=0}^{k}(-1)^{j'}\binom{k}{j}\int_{w_1=(n-d_1+j)\tau_1}^{c'(\boldsymbol{d})}\frac{(w_1-(n-d_1+j)\tau_1)^{r-1}}{\Gamma(r)} \\
&\Bigg[a_0H_{1}(w_1,0,\boldsymbol{d}_1,\boldsymbol{0},\boldsymbol{0})+\sum_{j=0}^Ja_j~H_{1}(w_1,0,\boldsymbol{d}_1,\boldsymbol{0},\boldsymbol{p}_j)+\underset{i\leq j}{\sum_{i=1}^J\sum_{j=1}^J}a_{ij}~H_{1}(w_1,0,\boldsymbol{d}_1,\boldsymbol{0},\boldsymbol{p}_{ij})~\Bigg]dw_1
\end{align*}

\end{document}